\documentclass[11pt,fleqn]{article}
\usepackage{amssymb,latexsym,amsmath,amsfonts,graphicx, ebezier}
\usepackage{pictex,color}
\usepackage{pict2e}

\advance\textheight by \topskip

\numberwithin{equation}{section}

\def\beq{\begin{equation}}
\def\eeq{\end{equation}}

\def\bit{\begin{itemize}}
\def\eit{\end{itemize}}

\makeatletter
\def\eqalign#1{\null\vcenter{\def\\{\cr}\openup\jot\m@th
\ialign{\strut$\displaystyle{##}$\hfil&$\displaystyle{{}##}$\hfil
\crcr#1\crcr}}\,}
\makeatother

\renewcommand{\Im}{\textup{Im}\,}
\renewcommand{\Re}{\textup{Re}\,}

\newcommand{\be}{\begin{equation}}
\newcommand{\ee}{\end{equation}}

\newcommand{\R}{\mathbb{R}}

\newcommand{\N}{\mathbb{N}}

    \def\bigO{{\cal O}}

    \def\P2n{{\rm P}_{{\rm II}}^{(n)}}

    \newtheorem{theorem}{Theorem}[section]
    
    \newtheorem{corollary}[theorem]{Corollary}
    \newtheorem{proposition}[theorem]{Proposition}
    
    \newtheorem{Definition}[theorem]{Definition}
    \newenvironment{definition}{\begin{Definition}\rm}{\end{Definition}}
    \newtheorem{Remark}[theorem]{Remark}
    \newenvironment{remark}{\begin{Remark}\rm}{\end{Remark}}
    \newtheorem{Example}[theorem]{Example}
    
    \newtheorem{Assumptions}[theorem]{Assumptions}

    \newenvironment{proof}%
    {\rm \trivlist \item[\hskip \labelsep{\bf Proof. }]}%
    {\hspace*{\fill}$\Box$\endtrivlist}
    {\rm \trivlist \item[\hskip \labelsep{\bf Proof}]}%
    {\hspace*{\fill}$\Box$\endtrivlist}

\begin{document}
\title{The generating function for the Bessel point process and a system of coupled Painlev\'{e} V equations}
\author{Christophe Charlier\footnote{Department of Mathematics, KTH Royal Institute of Technology, Lindstedtsv\"{a}gen 25, SE-114 28 Stockholm, Sweden. e-mail: cchar@kth.se}, Antoine Doeraene\footnote{Institut de Recherche en Math\'ematique et Physique,  Universit\'e
catholique de Louvain, Chemin du Cyclotron 2, B-1348
Louvain-La-Neuve, Belgium. e-mail: antoine.doeraene@uclouvain.be}
}

\maketitle

\begin{abstract} 
We study the joint probability generating function for $k$ occupancy numbers on disjoint intervals in the Bessel point process. This generating function can be expressed as a Fredholm determinant. We obtain an expression for it in terms of a system of coupled Painlev\'{e} V equations, which are derived from a Lax pair of a Riemann-Hilbert problem. This generalizes a result of Tracy and Widom \cite{TraWidLUE}, which corresponds to the case $k = 1$. We also provide some examples and applications. In particular, several relevant quantities can be expressed in terms of the generating function, like the gap probability on a union of disjoint bounded intervals, the gap between the two smallest particles, and large $n$ asymptotics for $n\times n$ Hankel determinants with a Laguerre weight possessing several jumps discontinuities near the hard edge.
\end{abstract}

\section{Introduction}

The Bessel point process is a determinantal point process on $\R^+$ arising as a limit point process of a wide range of mathematical models in random matrix theory \cite{For1,ForNag}. A celebrated toy example is the behaviour near $0$ of the squared singular values of Ginibre matrices, also known as the Laguerre Unitary Ensemble \cite{Vanlessen}. Other examples include non-intersecting squared Bessel paths \cite{KuijlaarsMartinezFinkelshteinWielonsky}, and the conditional Circular Unitary Ensemble near the edges \cite{ChCl2}.

\vspace{0.2cm}The main feature of determinantal point processes on a set $A \subseteq \R$ is that for all $n \in \N_{>0}$, the $n$-point correlation function $\rho_n : A^n \to \R$ is expressed in terms of a correlation kernel $K : A \times A \to \R$ as follows
\[\rho_n(x_1,...,x_n) = \det\left( K(x_j, x_\ell) \right)_{j,\ell = 1}^n.\]
In the Bessel point process, $A = \mathbb{R}^{+}$ and the kernel is given by
\begin{equation}
K^{\mathrm{Be}}(x,y) = \frac{J_{\alpha}(\sqrt{x})\sqrt{y}J_{\alpha}^{\prime}(\sqrt{y})-\sqrt{x}J_{\alpha}^{\prime}(\sqrt{x})J_{\alpha}(\sqrt{y})}{2(x-y)}, \qquad \alpha > -1,
\end{equation}
where $J_\alpha$ stands for the Bessel function of the first kind of order $\alpha$ (see \cite[formula 10.2.2]{NIST} for a definition of $J_{\alpha}$).

Important quantities related to point processes are \emph{occupancy numbers}. Given a Borel set $B \subseteq \R^+$, the occupancy number $n_B$ is the random variable defined as the number of particles that fall into $B$. Determinantal point processes are always locally finite, i.e. $n_B$ is finite with probability $1$ for $B$ bounded. Moreover, all particles are distinct with probability $1$. In particular, it allows us to enumerate particles for the Bessel point process in the following way,
\[0 < \zeta_1 < \zeta_2 < \zeta_3 < ...\]
In this paper, we focus on the joint behaviour of a finite number of particles, which can be completely understood via the joint probability generating function of the occupancy numbers of some particular sets. Let $k \in \N_{>0}$, $\vec{s}=(s_1,...,s_k) \in \mathbb{C}^{k}$ and $\vec{x} = (x_{1},...,x_{k}) \in (\mathbb{R}^{+})^{k}$ be such that $0 = x_0 < x_1 < x_2 < ... < x_k < +\infty$. We will be interested in the function
\begin{equation}\label{def of F}
F(\vec{x},\vec{s}) = \mathbb{E}\Bigg( \prod_{j = 1}^k s_j^{n_{(x_{j-1}, x_j)}} \Bigg) = \sum_{m_{1},...,m_{k} \geq 0} \mathbb{P}\Bigg(\bigcap_{j=1}^{k}n_{(x_{j-1},x_{j})}=m_{j}\Bigg)\prod_{j=1}^{k} s_{j}^{m_{j}}.
\end{equation}
It is known \cite[Theorem 2]{Soshnikov} that $F(\vec{x},\vec{s})$ is an entire function in $s_{1},...,s_{k}$ and can be expressed as a Fredholm determinant as follows
\begin{equation}\label{F Fredholm}
F(\vec{x},\vec{s}) = \det \Bigg( 1- \chi_{(0,x_{k})}\sum_{j=1}^{k}(1-s_{j})\mathcal{K}^{\mathrm{Be}}\chi_{(x_{j-1},x_{j})} \Bigg),
\end{equation}
where $\mathcal{K}^{\mathrm{Be}}$ denotes the integral operator acting on $L^2(\R^+)$ whose kernel is the Bessel kernel $K^{\mathrm{Be}}$, and where $\chi_A$ is the projection operator onto $L^2(A)$. 

\vspace{0.2cm} The goal of this paper is to express $F(\vec{x},\vec{s})$ explicitly in terms of $k$ functions which satisfy a system of $k$ coupled Painlev\'{e} V equations. Analogous generating functions for the Airy point process have been recently studied in \cite{ClaeysDoeraene} (for a general $k\in \mathbb{N}_{>0}$), and in \cite{XuDai} (for the case $k = 2$ with an extra root-type singularity). In both cases, the authors expressed it in terms of a system of coupled Painlev\'{e} II equations.


\subsection*{Tracy-Widom formula for $k = 1$}

In \cite{TraWidLUE}, Tracy and Widom have studied $F(x_{1},s_{1})$, i.e. the case $k = 1$. This is the probability generating function of $n_{(0,x_{1})}$. In particular, we can deduce from $F(x_{1},s_{1})$ the probability distribution of the $\ell$-th smallest particle $\zeta_{\ell}$ as follows
\begin{equation}
\mathbb{P}(\zeta_\ell > x_{1}) = \mathbb{P}\left(n_{(0,x_{1})} < \ell\right) = \sum_{j = 0}^{\ell-1} \left.\frac{1}{j!} \partial_{s_{1}}^{j} F(x_{1},s_{1}) \right|_{s_{1} = 0}.
\end{equation}
Their theorem states that for $0 \leq s_{1} < 1$ and $x_{1} > 0$,
\begin{equation}\label{eq: tracy widom formula}
F(x_{1}, s_{1}) = \exp \left( -\frac{1}{4} \int_{0}^{x_{1}} \log \left(  \frac{x_{1}}{\xi}\right)q^{2}(\xi;s_{1})d\xi \right),
\end{equation}
where $q(\xi;s_{1})$ satisfies the Painlev\'{e} V equation given by
\begin{equation}\label{eq: painleve V}
\xi q\big(1-q^{2}\big)\big(\xi qq^{\prime}\big)^{\prime} + \xi \big(1-q^{2}\big)^{2}\left( (\xi q^{\prime})^{\prime} + \frac{q}{4} \right) + \xi^{2} q \left( qq^{\prime} \right)^{2} = \alpha^{2} \frac{q}{4},
\end{equation}
with boundary condition $q(\xi;s_{1}) \sim \sqrt{1-s_{1}}J_{\alpha}(\sqrt{\xi})$ as $\xi \to 0$, and where primes denote derivatives with respect to $\xi$.

\subsection*{Joint distribution for $k$ particles}

Let us start with the case $k = 2$ for simplicity. Let $m_{1},m_{2} \in \mathbb{N}_{>0}$ be such that $m_{1}<m_{2}$. If $0< x_1 < x_2<+\infty$, the joint distribution of the $m_{1}$-th and $m_{2}$-th smallest particles in the Bessel point process is given in terms of $F((x_{1},x_{2}),(s_{1},s_{2}))$ by
\begin{equation}\label{lol 13}
\begin{array}{r c l}
\displaystyle \mathbb{P}\Big(\zeta_{m_1} > x_1, \zeta_{m_2} > x_2\Big) & =  & \displaystyle \sum_{\substack{j_1 < m_1\\j_1+j_2 < m_2}} \mathbb{P}\Big( n_{(0,x_1)} = j_1, n_{(x_1, x_2)} = j_2 \Big) \\
& = & \displaystyle \sum_{\substack{j_1 < m_1\\j_1+j_2 < m_2}} \frac{1}{j_1!j_2!} \left. \partial_{s_{1}}^{j_{1}}\partial_{s_{2}}^{j_{2}} F((x_1,x_2),(s_1,s_2)) \right|_{s_1 = s_2 = 0}.
\end{array}
\end{equation}
More generally, for $k \in \mathbb{N}_{>0}$, the function $F$ can be used to express the joint probability distribution of any $k$ distinct particles. The general formula for $k > 2$ can be easily generalized from \eqref{lol 13}. Let $m_{1},...,m_{k} \in \mathbb{N}_{>0}$ and $\vec{x} = (x_{1},...,x_{k}) \in (\mathbb{R}^{+})^{k}$ be such that $m_1 < m_2 < ... < m_k$ and $ x_1 < ... < x_k$. We have
\begin{equation}\label{probgenk}
\begin{array}{r c l}
\displaystyle \mathbb{P}\left(\cap_{j=1}^k\left(\zeta_{m_j} > x_j\right)\right) & = & \displaystyle \sum \mathbb P\Big(\bigcap_{j=1}^k\left( n_{(x_{j-1}, x_j)} = m_j\right)\Big)\\[0.3cm]
 & = & \displaystyle \sum \frac{1}{j_1!j_2!\ldots j_k!} \left. \partial_{s_{1}}^{j_{1}}\partial_{s_{2}}^{j_{2}}\dots \partial_{s_{k}}^{j_{k}} F(\vec x,\vec s)\right|_{\vec s = 0},
\end{array}
\end{equation}
where the sum is taken over all indices $j_1,\ldots, j_k$ such that 
\vspace{-0.2cm}\begin{equation*}
j_1<m_1,\quad j_1+j_2<m_2,\quad  \ldots\quad \sum_{i=1}^k j_i<m_k.
\end{equation*}

\vspace{-0.2cm}\hspace{-0.6cm}We give other quantities of interest which can be expressed in terms of $F$ in Section \ref{Section: examples and applications}.

\subsection*{Tracy-Widom type formula for $F$}

The Tracy-Widom formula \eqref{eq: tracy widom formula} characterized $F$ in the case $k = 1$ in terms of a function $q$ which satisfies the Painlev\'{e} V equation \eqref{eq: painleve V}. The main result of this paper gives a generalisation of that for a general $k\in \mathbb{N}_{>0}$. We find that $F$ can be expressed in terms of $k$ functions $q_{1}$,...,$q_{k}$, which satisfy a system of $k$ coupled Painlev\'e V equations with Bessel boundary conditions at $0$. The theorem reads as follows.

\begin{theorem}\label{thm: main thm}
Let $\vec{r} = (r_{1},...,r_{k}) \in (\mathbb{R}^{+})^{k}$ and $\vec{s}=(s_{1},...,s_{k}) \in [0,1]^{k}$ be such that
\begin{align}
& r_{j}>r_{j-1}, \mbox{ for } j=1,...,k, \mbox{ where } r_{0}:=0,\label{condition r}\\
& s_{j} \neq s_{j+1}, \mbox{ for } j=1,...,k, \mbox{ where } s_{k+1}:=1. \label{condition s} 
\end{align}
For $x>0$, the joint probability generating function $F(\vec{r}x,\vec{s})$ is given by
\begin{equation}\label{eq: tracy widom type general k}
F(\vec{r}x,\vec{s}) = \prod_{j=1}^{k} \exp \left( -\frac{r_{j}}{4}\int_{0}^{x} \log\left( \frac{x}{\xi} \right) q_{j}^{2}(\xi;\vec{r},\vec{s})d\xi  \right),
\end{equation}
where the functions $q_{1}(\xi;\vec{r},\vec{s})$,...,$q_{k}(\xi;\vec{r},\vec{s})$ satisfy the system of $k$ equations given by
\begin{equation}\label{system of ODEs for the q_j}
\xi q_{j}\bigg(1-\sum_{\ell=1}^{k}q_{\ell}^{2}\bigg)\sum_{\ell=1}^{k}\big(\xi q_{\ell} q_{\ell}^{\prime}\big)^{\prime} + \xi \bigg(1-\sum_{\ell=1}^{k}q_{\ell}^{2}\bigg)^{2}\left( (\xi q_{j}^{\prime})^{\prime} + \frac{r_{j}q_{j}}{4} \right) + \xi^{2} q_{j} \bigg( \sum_{\ell=1}^{k}q_{\ell}q_{\ell}^{\prime} \bigg)^{2} = \alpha^{2} \frac{q_{j}}{4},
\end{equation}
where $j = 1,2,...,k$, and where primes denote derivatives with respect to $\xi$. Furthermore, for every $j \in \{1,...,k\}$, $q_{j}^{2}(\xi;\vec{r},\vec{s})$ is real for $\xi>0$ and satisfies the boundary condition 
\begin{equation}\label{boundary conditions for the q_j}
q_{j}(\xi;\vec{r},\vec{s}) = \sqrt{s_{j+1}-s_{j}} J_{\alpha}(\sqrt{r_{j}\xi})(1+\bigO(\xi)), \qquad \mbox{ as } \xi \to 0.
\end{equation}
\end{theorem}
\begin{remark}
Theorem \ref{thm: main thm} is a generalization for $k \in \mathbb{N}_{>0}$ of the Tracy-Widom formula. Indeed, if $k = 1$, $x = x_{1}$ and $r_{1} = 1$, the above formulas \eqref{eq: tracy widom type general k} and \eqref{system of ODEs for the q_j} are reduced to \eqref{eq: tracy widom formula} and \eqref{eq: painleve V}, and $q_{1}$ given in Theorem \ref{thm: main thm} and $q$ given by \eqref{eq: painleve V} satisfy the same boundary condition at $0$.
\end{remark}
\begin{remark}
The system \eqref{system of ODEs for the q_j} with boundary conditions \eqref{boundary conditions for the q_j} given in Theorem \ref{thm: main thm} has at least one solution $(q_{1},...,q_{k})$, but there is no guarantee that this solution is unique. Therefore, the functions $q_{1},...,q_{k}$ that appear in \eqref{eq: tracy widom type general k} are not defined through the system \eqref{system of ODEs for the q_j}, but they are explicitly constructed from the solution $\Phi$ of a Riemann-Hilbert (RH) problem. This RH problem is presented in Section \ref{sec: model rh problem}.
\end{remark}
The asymptotic behaviour \eqref{boundary conditions for the q_j} allows to compute directly the small $x$ asymptotics of $F(\vec{r}x,\vec{s})$, and is given in the following corollary.
\begin{corollary}
Let $x > 0$, fix $\vec{r} = (r_{1}$,...,$r_{k}) \in (\mathbb{R}^{+})^{k}$ independent of $x$ such that $r_{1}<...<r_{k}$, and fix $\vec{s} = (s_{1},...,s_{k}) \in [0,1]^{k}$  independent of $x$ such that $s_{j} \neq s_{j+1}$ for $j = 1,...,k$ with $s_{k+1} = 1$. We have
\begin{equation}
F(\vec{r}x,\vec{s}) = 1 - \sum_{j=1}^{k} (s_{j+1}-s_{j}) J_{\alpha+1}(\sqrt{r_{j}x})^{2} + \bigO(x^{2+\alpha}), \qquad \mbox{ as } x \to 0.
\end{equation}
\end{corollary}
\begin{proof}
This a direct consequence of Theorem \ref{thm: main thm} together with the formula $z\Gamma(z) = \Gamma(z+1)$ and the limiting behaviour $J_{\alpha}(x) = \left( \frac{x}{2} \right)^{\alpha} \frac{1}{\Gamma(\alpha+1)}\big(1+\bigO\big(x^{2}\big)\big)$ as $x \to 0$ (see \cite[formula 10.7.3]{NIST}).
\end{proof}

\subsection*{Asymptotics for $q_{1},...,q_{k}$ as $s_{j} \to s_{j+1}$ or as $r_{j} \to r_{j-1}$}

In Theorem \ref{thm: main thm}, it is essential that the conditions \eqref{condition r} and \eqref{condition s} hold. Suppose that one of these conditions is not satisfied, i.e. suppose we have $s_{j} = s_{j+1}$ or $r_{j} = r_{j-1}$ for a certain $j \in \{1,...,k\}$. Then, from \eqref{def of F}, we have
\begin{equation}\label{lol 14}
F(\vec{r}x,\vec{s}) = F(\vec{r}^{[j]}x,\vec{s}^{[j]}),
\end{equation}
where for a given vector $\vec{w} = (w_1,...,w_k)$, we use the notation  $\vec{w}^{[j]}$ for the vector $\vec{w}$ with its $j$-th component removed, i.e. $\vec{w}^{[j]} = (w_1,...,w_{j-1}, w_{j+1}, ... w_k)$. Theorem \ref{thm: main thm} applied to the right-hand side of \eqref{lol 14} allows to rewrite $F(\vec{r}x,\vec{s})$ in terms of a solution of a system of $k-1$ coupled Painlev\'{e} equations. Thus, if $\vec{r}$ and $\vec{s}$ satisfy \eqref{condition r} and \eqref{condition s}, as $s_j \to s_{j+1}$ or $r_j \to r_{j-1}$ and $\xi$ fixed, we should observe $||\vec{q}(\xi;\vec{r},\vec{s})-\vec{q}^{[j]}(\xi;\vec{r}^{[j]},\vec{s}^{[j]})||\to 0$, where $\vec{q}=(q_{1},...,q_{k})$. Theorem \ref{thm: residual thm} below gives such asymptotics in the above degenerate cases.

\newpage
\begin{theorem}\label{thm: residual thm}
Fix $x>0$ and let $\vec{r} = (r_{1},...,r_{k}) \in (\mathbb{R}^{+})^{k}$ and $\vec{s}=(s_{1},...,s_{k}) \in [0,1]^{k}$ be such that \eqref{condition r} and \eqref{condition s} are satisfied.
\begin{itemize}
\item[1.] Let $j \in \{1,...,k\}$. As $s_{j} \to s_{j+1}$, we have
\begin{align}
& q_{j}^{2}(x;\vec{r},\vec{s}) = \bigO(|s_{j}-s_{j+1}|), \\
& |q_{\ell}^{2}(x;\vec{r},\vec{s}) - q_{\tilde{\ell}}^{2}(x;\vec{r}^{[j]},\vec{s}^{[j]})| = \bigO(|s_{j}-s_{j+1}|), \label{lol 16}
\end{align}
and \eqref{lol 16} holds for any $\ell \neq j$, and where $\tilde{\ell} = \ell$ if $\ell < j-1$ and $\tilde{\ell} = \ell-1$ if $\ell > j$.
\item[2.] Let $j \in \{2,...,k\}$. As $r_{j} \to r_{j-1}$ and if $s_{j+1} \neq s_{j-1}$, we have
\begin{align}
& q_{j-1}^{2}(x;\vec{r},\vec{s}) = \frac{s_{j}-s_{j-1}}{s_{j+1}-s_{j-1}} q_{j-1}^{2}(x;\vec{r}^{[j]},\vec{s}^{[j]})+\bigO(r_{j}-r_{j-1}), \\
& q_{j}^{2}(x;\vec{r},\vec{s}) = \frac{s_{j+1}-s_{j}}{s_{j+1}-s_{j-1}}q_{j-1}^{2}(x;\vec{r}^{[j]},\vec{s}^{[j]})+\bigO(r_{j}-r_{j-1}), \\
& |q_{\ell}^{2}(x;\vec{r},\vec{s}) - q_{\tilde{\ell}}^{2}(x;\vec{r}^{[j]},\vec{s}^{[j]})| = \bigO(r_{j}-r_{j-1}), \label{lol 15}
\end{align}
and \eqref{lol 15} holds for any $\ell \neq j-
1$, $\ell \neq j$, and where $\tilde{\ell} = \ell$ if $\ell < j-2$ and $\tilde{\ell} = \ell-1$ if $\ell > j$.
\item[3.] As $r_{1} \to 0$, we have
\begin{align}
& q_{1}^{2}(x;\vec{r},\vec{s}) = \bigO(r_{1}^{\alpha}), \\
& |q_{\ell}^{2}(x;\vec{r},\vec{s}) - q_{\tilde{\ell}}^{2}(x;\vec{r}^{[j]},\vec{s}^{[j]})| = \bigO(r_{1}), \label{lol 20}
\end{align}
and \eqref{lol 20} holds for any $\ell \geq 2$, and where $\tilde{\ell} = \ell-1$.
\end{itemize}
\end{theorem}

\subsection*{Outline}
We provide some examples and applications of our results in Section \ref{Section: examples and applications}. The system of $k$ coupled Painlev\'{e} V equations for $q_{1},...,q_{k}$, given by \eqref{system of ODEs for the q_j}, is obtained from a Lax pair of a model Riemann-Hilbert (RH) problem which is introduced in Section \ref{sec: model rh problem}, and whose solution is denoted $\Phi$. In Section \ref{Section: proof of main thm}, using the procedure introduced by Its, Izergin, Korepin and Slavnov \cite{IIKS} for integrable operators, we relate $F$ with $\Phi$ through a differential identity, which we integrate to prove Theorem \ref{thm: main thm}. In Section \ref{Section: small x asymptotics}, we perform the Deift/Zhou \cite{DeiftZhou1992, DeiftZhou} steepest descent method on the model RH problem to obtain the asymptotic behaviour of $q_{j}(x)$ as $x \to 0$. The first and second part of Theorem \ref{thm: residual thm} are obtained in Section \ref{Section: asymptotics as s_j to s_j+1} and Section \ref{Section: Asymptotics as r_j to r_j-1} respectively, via a more direct steepest descent method on the RH problem.

\section{Examples and applications}
\label{Section: examples and applications}
The applications presented in this section are similar to those presented in \cite{ClaeysDoeraene} (for the Airy point process), and are adapted here for the Bessel point process.
\subsection{Gap probability on a union of disjoint bounded intervals}

A gap in a point process is the event of finding no particle in a certain set. The Tracy-Widom distribution given by \eqref{eq: tracy widom formula} corresponds to the gap probability for an interval of the form $(0,a)$, where $0<a<+\infty$, and can be rewritten as
\begin{equation}
\mathbb{P}\Big( n_{(0,a)} = 0 \Big) = F(a,0) = \exp \left( \frac{a}{4} \int_0^1 \log (\xi) q_1^2(\xi) d\xi \right),
\end{equation}
where in the above expression we have used the definition of $F$ given by \eqref{def of F} for the first equality, and where we have applied Theorem \ref{thm: main thm} (with $k = 1$, $x=1$, $r_{1} = a$, $s_{1} = 0$) for the second equality. The gap probability in the Bessel point process for a single interval of the form $(a,b)$, with $0<a<b<+\infty$, is given by
\begin{equation*}
\begin{array}{r c l}
\displaystyle \mathbb{P}\Big( n_{(a,b)} = 0 \Big) & = & \displaystyle F((a,b),(1,0)) \\
 & = & \displaystyle \exp\left( \frac{a}{4} \int_0^1 \log (\xi) q_1^2(\xi) d\xi \right) \exp\left( \frac{b}{4} \int_0^1 \log(\xi) q_2^2(\xi) d\xi \right),
\end{array}
\end{equation*}
where we have used Theorem \ref{thm: main thm} (with $k=2$, $x = 1$, $\vec{r} = (a,b)$, $\vec{s} = (1,0)$) for the second equality.

\vspace{0.2cm}This computation can be generalized for the gap probability of any finite union of disjoint bounded intervals. Let $\ell \in \mathbb{N}_{>0}$ be the number of intervals and $0 < a_1 < b_1 < a_2 < ... < b_\ell < +\infty$, we have
\begin{equation*}
\begin{array}{r c l}
\displaystyle \mathbb{P}\Bigg( \bigcap_{j = 1}^\ell n_{(a_j, b_j)} = 0 \Bigg) & = & \displaystyle F((a_1, b_1, ...,a_{\ell},b_\ell), (1,0,...,1,0)) \\
& = & \displaystyle \exp\Bigg( \frac{1}{4} \int_0^1 \log (\xi) \sum_{j=1}^{\ell} (a_{j}q_{2j-1}^{2}(\xi)+b_{j}q_{2j}^{2}(\xi)) d\xi \Bigg),
\end{array}
\end{equation*}
where we have applied Theorem \ref{thm: main thm} with $k = 2\ell$, $x=1$, $\vec{r} = (a_{1},b_{1},...,a_{\ell},b_{\ell})$ and $\vec{s} = (1,0,...,1,0)$, and where the $2 \ell$ functions $q_{1},...,q_{2\ell}$ satisfy the system \eqref{system of ODEs for the q_j}.


\subsection{Distribution of the smallest particle in the thinned and conditional Bessel point process}
The generating function $F$ is also useful in the context of \emph{thinning}. The thinning of a determinantal point process is a procedure introduced by Bohigas and Pato \cite{BohigasPato1,BohigasPato2} that consists in building a new point process by removing each particle with a certain probability.

\vspace{0.2cm}We consider a constant and independent thinning of the Bessel point process. Given a realization $0 < \zeta_1 < \zeta_2 < ...$, it consists of removing each of these particles independently with the same probability $s \in (0,1)$. The thinned point process is composed of the remaining particles $0<\xi_1 < \xi_2 < ...$, and is again a determinantal point process, whose correlation kernel is given by $(1-s)K^{\textrm{Be}}$ (see \cite{LMR}). For a given Borel set $B \subset \mathbb{R}^{+}$, we denote $\tilde{n}_{B}$ for the occupancy number of $B$ in the thinned point process. The probability distribution of $\xi_1$ (smallest particle of the thinned point process) can be deduced from $F$ with $k = 1$, since by \eqref{def of F} we have
\begin{equation}\label{eq: xi distribution}
\hspace{-0.5cm}\mathbb{P}(\xi_1 > x) = \sum_{j = 0}^{+\infty} \mathbb{P}\Big( n_{(0,x)} = j ~ \cap \tilde{n}_{(0,x)}=0 \Big) = \sum_{j = 0}^{+\infty} \mathbb{P}\Big( n_{(0,x)} = j \Big) s^j = F(x,s),
\end{equation}
and where $F(x,s)$ admits the Tracy-Widom formula \eqref{eq: tracy widom formula}.
We can also consider another situation, where we have information about the thinned point process. Suppose that we observe the event $\tilde{n}_{(0,x_{2})}=0$ for a certain $x_{2}>0$ (we \textit{condition} on this event), and from there, we want to retrieve information on $\zeta_{1}$. The distribution of $\left. \zeta_{1}\right|_{\tilde{n}_{(0,x_{2})}=0}$ (the smallest particle in the conditional point process) is given by
\begin{equation}
\mathbb{P}\Big( \left. \zeta_{1}\right|_{\tilde{n}_{(0,x_{2})}=0}>x_{1} \Big) = \mathbb{P}(\zeta_1 > x_1 ~|~ \xi_1 > x_2) = \frac{\mathbb{P}(\zeta_1 > x_1 \cap \xi_1 > x_2)}{\mathbb{P}(\xi_1 > x_2)},
\end{equation}
where $0 < x_1 < x_2$. The denominator in the above expression is just given by $F(x_{2},s)$, as shown in \eqref{eq: xi distribution}. The numerator is slightly more involved, and can be expressed in terms of $F$ with $k = 2$ as follows
\begin{equation}
\mathbb{P}(\zeta_1 > x_1 \cap \xi_1 > x_2) = \sum_{j = 0}^{+\infty} s^j \mathbb{P}\Big( n_{(0,x_1)} = 0 \cap n_{(x_1,x_2)} = j \Big) = F( (x_1,x_2), (0, s) ).
\end{equation}
Therefore, Theorem \ref{thm: main thm} allows us to express the distribution of the smallest particle in the conditional point process as
\begin{equation}
\mathbb{P}\Big( \left. \zeta_{1}\right|_{\tilde{n}_{(0,x_{2})}=0}>x_{1} \Big) = \exp \left( \frac{1}{4}\int_{0}^{1}  \left[ x_{1}q_{1}(\xi) + x_{2}(q_{2}(\xi)-\tilde{q}(\xi)) \right]\log \xi d\xi \right),
\end{equation}
where $q_{1}$, $q_{2}$ satisfy the system \eqref{system of ODEs for the q_j} with $k = 2$, $x=1$, $\vec{r} = (x_{1},x_{2})$, $\vec{s} = (0,s)$ and $\tilde{q}$ satisfies \eqref{system of ODEs for the q_j} with $k = 1$, $x = 1$, $r_{1} = x_{2}$ and $s_{1} = s$.

\subsection{Smallest LUE eigenvalues}
\label{Section: convergence of LUE to Bessel}
The Bessel point process appears as a limiting point process for eigenvalues of random matrices whose spectrum possesses a hard edge. The most well-known example is the Laguerre Unitary Ensemble (LUE), which is the set of $n \times n$ positive definite Hermitian matrices $M$ endowed with the probability measure
\begin{equation}\label{prob measure LUE}
\frac{1}{\widetilde{Z}_{n,\alpha}}(\det M)^{\alpha} e^{-\mathrm{Tr} M}dM, \qquad dM = \prod_{j=1}^{n}dM_{ii} \prod_{1\leq i<j \leq n} d\Re M_{ij} d\Im M_{ij},
\end{equation}
where $\widetilde{Z}_{n,\alpha}$ is the normalization constant. Since the matrix $M$ is positive definite, its eigenvalues $\lambda_{1},...,\lambda_{n}$ are positive and $0$ is a hard edge of the spectrum. By integrating over the unitary group the probability measure \eqref{prob measure LUE}, it reduces to the probability measure on $(\mathbb{R}^{+})^{n}$ given by
\begin{equation}\label{distri eig LUE}
\frac{1}{n! Z_{n,\alpha}}\prod_{1\leq i < j \leq n}(\lambda_{j}-\lambda_{i})^{2} \prod_{j=1}^{n}\lambda_{j}^{\alpha}e^{-\lambda_{j}}d\lambda_{j},
\end{equation}
where $Z_{n,\alpha}$ is the partition function.  It is well-known \cite{Deift} that \eqref{distri eig LUE} is a determinantal point process whose correlation kernel is
\begin{equation}
K_{n}^{\mathrm{LUE}}(\lambda,\nu) = \sqrt{w(\lambda)w(\nu)}\sum_{j=0}^{n-1} p_{j}(\lambda)p_{j}(\nu), \qquad \lambda, \nu > 0,
\end{equation}
where $w(x) = x^{\alpha}e^{-x}$ and $p_{j}(x)$ is the Laguerre orthonormal polynomial of degree $j$, i.e. it satisfies
\begin{equation}
\int_{0}^{\infty} p_{j}(x)p_{\ell}(x)w(x)dx = \delta_{j\ell}, \qquad \mbox{ for } \ell = 0,1,2,...,j.
\end{equation}
Near the hard edge, the LUE kernel converges to the Bessel kernel as $n \to \infty$. More precisely, after the rescaling 
\begin{equation}\label{rescaling LUE}
x_{j} = 4 n \lambda_{j} \qquad \mbox{ for }j = 1,...,n,
\end{equation}
the following limit holds
\begin{equation}\label{convergence of LUE kernel}
\lim_{n\to \infty} \frac{1}{4n} K_{n}^{\mathrm{LUE}}\left( \frac{x}{4n},\frac{y}{4n} \right) = K^{\mathrm{Be}}(x,y).
\end{equation}
This limit implies also trace-norm convergence of the associated operator when acting on bounded intervals. Therefore, after the proper rescaling between $\vec{\lambda}$ and $\vec{x}$ given by \eqref{rescaling LUE}, we have
\begin{equation}\label{convergence generating Laguerre}
F_{n}^{\mathrm{LUE}}(\vec{\lambda},\vec{s}) := \det \Bigg( 1- \chi_{(0,\lambda_{k})}\sum_{j=1}^{k}(1-s_{j})\mathcal{K}_{n}^{\mathrm{LUE}}\chi_{(\lambda_{j-1},\lambda_{j})} \Bigg) = F(\vec{x},\vec{s}) + o(1)
\end{equation}
as $n \to \infty$, and where $\lambda_{0} := 0$ and $\mathcal{K}_{n}^{\mathrm{LUE}}$ is the integral operator whose kernel is $K_{n}^{\mathrm{LUE}}$. On the other hand, $F_{n}^{\mathrm{LUE}}(\vec{\lambda},\vec{s})$ can also be written as the following ratio of Hankel determinants
\begin{equation}\label{ratio Hankel determinants}
F_{n}^{\mathrm{LUE}}(\vec{\lambda},\vec{s}) = \frac{\displaystyle \det \Bigg( \int_{0}^{\infty} w(x)\Bigg(1-\sum_{j=1}^{k}(1-s_{j})\chi_{(\lambda_{j-1},\lambda_{j})}(x)\Bigg) x^{i+j-2}dx \Bigg)_{i,j=1}^{n}}{\displaystyle \det \left( \int_{0}^{\infty} w(x)x^{i+j-2}dx \right)_{i,j=1}^{n}},
\end{equation}
where the denominator of the above expression is the partition function of the LUE and is well-known (see \cite[formula 17.6.5]{Mehta}). In particular, Theorem \ref{thm: main thm} together with \eqref{convergence generating Laguerre} implies large $n$ asymptotics for the ratio \eqref{ratio Hankel determinants} up to constant term, but this does not provide an estimate for error term $o(1)$ in \eqref{convergence generating Laguerre}. 

\subsection{Ratio probability between the two smallest particles}
Two quantities of interest are the ratio and the gap probabilities between the two smallest particles in the Bessel point process, namely $Q_{\alpha}(r) = \mathbb{P}\left( \frac{\zeta_{2}}{\zeta_{1}}>r \right)$ and $G_{\alpha}(d) = \mathbb{P}\left( \zeta_{2}-\zeta_{1} > d \right)$, where $r>1$ is the size of the ratio and $d>0$ is the size of the gap. The ratio probability was obtained in \cite{AtkChaZoh} and the gap probability in \cite{ForWit}. Note that Theorem \ref{thm: main thm} expresses quantities related to ratios of particles  more naturally than quantities related to differences of particles. Indeed, if we choose $\vec{r} = (1,r_{2},...,r_{k})$ (i.e. $r_{1} = 1$) in \eqref{eq: tracy widom type general k}, the numbers $r_{2}$,...,$r_{k}$ are related to the ratios $\frac{\zeta_{2}}{\zeta_{1}}$,...$\frac{\zeta_{k}}{\zeta_{1}}$. In this section, we start by expressing $Q_{\alpha}(r)$ in terms of $F$. By definition, we have
\begin{equation}
\begin{array}{r c l}
\displaystyle Q_{\alpha}(r) & = & \displaystyle \int_{0}^{\infty} \partial_{\xi} \mathbb{P}\left.\left( \zeta_{1} \leq \xi  \cap \zeta_{2}>rx \right)\right|_{\xi = x}dx, \\[0.3cm]
& = & \displaystyle \int_{0}^{\infty} \partial_{\xi} \left. \mathbb{P}\left( n_{(0,\xi)}=1 \cap n_{(\xi,rx)} = 0 \right)\right|_{\xi = x}dx.
\end{array}
\end{equation}
The probability in the integrand can be obtained from the generating function \eqref{def of F} as follows
\begin{equation}
\partial_{s} \left.F((\xi,rx),(s,0))\right|_{s=0} = \mathbb{P}\left( n_{(0,\xi)}=1 \cap n_{(\xi,rx)} = 0 \right),
\end{equation}
and thus
\begin{equation*}
\hspace{-0.1cm}\begin{array}{r c l}
\displaystyle Q_{\alpha}(r) & = & \displaystyle \int_{0}^{\infty} \partial_{\xi} \left. \partial_{s} \left.F((\xi,rx),(s,0))\right|_{s=0}\right|_{\xi = x}dx \\[0.3cm]
& = & \displaystyle \int_{0}^{\infty}  \partial_{r_{1}} \left. \partial_{s} \left.F((r_{1}x,rx),(s,0))\right|_{s=0}\right|_{r_{1} = 1}\frac{dx}{x}, \\
& = & \displaystyle \int_{0}^{\infty}  \partial_{r_{1}} \left. \partial_{s} \left. \exp \left( -\frac{1}{4}\int_{0}^{x}\left[ r_{1}q_{1}^{2}(\xi)+rq_{2}^{2}(\xi) \right] \log \left( \frac{x}{\xi} \right)d\xi \right)\right|_{s=0}\right|_{r_{1} = 1}\frac{dx}{x},
\end{array}
\end{equation*}
where we have applied Theorem \ref{thm: main thm} with $k=2$, $\vec{r} = (r_{1},r)$, $\vec{s}=(s,0)$. It is worth comparing this formula with the result obtained in \cite[Theorem 1.7]{AtkChaZoh}, which is given by
\begin{equation}
Q_{\alpha}(r) = \frac{1}{4^{\alpha+1}\Gamma(1+\alpha)\Gamma(2+\alpha)}\int_{0}^{\infty} x^{\alpha} e^{I(x;r)}dx,
\end{equation}
where 
\begin{equation}\label{def of I}
I(x;r) = \displaystyle -\frac{1}{4} \int_{0}^{x} (\widetilde{q}_{1}^{2}(\xi;r)+r\widetilde{q}_{2}^{2}(\xi;r))\log \left( \frac{x}{\xi}\right) d\xi.
\end{equation}
The functions $\widetilde{q}_{1}^{2}(\xi;r)$ and $\widetilde{q}_{2}^{2}(\xi;r)$ are real and analytic for $\xi \in (0,\infty)$ and $r \in (1,\infty)$, and they satisfy the following system of two coupled Painlev\'{e} V equations:
\begin{align}
& \hspace{-0.95cm} \xi \widetilde{q}_{1}\bigg(1 \hspace{-0.08cm} -\hspace{-0.08cm} \sum_{j=1}^{2}\widetilde{q}_{j}^{2}\bigg) \sum_{j=1}^{2}(\xi \widetilde{q}_{j}\widetilde{q}_{j}^{\prime})^{\prime} + \bigg[ \xi \left( (\xi \widetilde{q}_{1}^{\prime})^{\prime} + \frac{\widetilde{q}_{1}}{4}\right) + \frac{1}{\widetilde{q}_{1}^{3}} \bigg]\bigg(1\hspace{-0.08cm}-\hspace{-0.08cm}\sum_{j=1}^{2}\widetilde{q}_{j}^{2}\bigg)^{\hspace{-0.08cm}2} \hspace{-0.08cm}+ \xi^{2}\widetilde{q}_{1}\bigg( \sum_{j=1}^{2}\widetilde{q}_{j}\widetilde{q}_{j}^{\prime} \bigg)^{2} \hspace{-0.08cm} = \frac{\alpha^{2} \widetilde{q}_{1}}{4}, \nonumber \\
& \hspace{-0.95cm} \xi \widetilde{q}_{2}\bigg(1\hspace{-0.08cm} -\hspace{-0.08cm} \sum_{j=1}^{2}\widetilde{q}_{j}^{2}\bigg) \sum_{j=1}^{2}(\xi \widetilde{q}_{j}\widetilde{q}_{j}^{\prime})^{\prime} + \xi \left( (\xi \widetilde{q}_{2}^{\prime})^{\prime} + \frac{r\widetilde{q}_{2}}{4}\right) \bigg(1\hspace{-0.08cm}-\hspace{-0.08cm}\sum_{j=1}^{2}\widetilde{q}_{j}^{2}\bigg)^{\hspace{-0.08cm}2} \hspace{-0.08cm} + \xi^{2}\widetilde{q}_{2}\bigg( \sum_{j=1}^{2}\widetilde{q}_{j}\widetilde{q}_{j}^{\prime} \bigg)^{2} \hspace{-0.08cm} = \frac{\alpha^{2} \widetilde{q}_{2}}{4}, \label{system for q_1 and q_2 tilde}
\end{align}
where primes denote derivatives with respect to  $\xi$. Furthermore, the functions $\widetilde{q}_{1}$ and $\widetilde{q}_{2}$ satisfy the following boundary conditions: as $\xi \to 0$, we have
\begin{align}
& \widetilde{q}_{1}(\xi) = \sqrt{\frac{2}{\alpha+2}}(1+\bigO(\xi)), \label{small x asymp q_1 tilde} \\
& \widetilde{q}_{2}(\xi) = (1-r^{-1})J_{\alpha+2}(\sqrt{r\xi})(1+\bigO(\xi)) = \frac{(1-r^{-1})(r\xi)^{\frac{\alpha+2}{2}}}{2^{\alpha+2}\Gamma(\alpha+3)}(1+\bigO(\xi)). \nonumber
\end{align}
The authors obtained also other asymptotics for $\widetilde{q}_{1}(\xi;r)$ and $\widetilde{q}_{2}(\xi;r)$ in various regimes of $r$ and $x$ (see \cite[Theorem 1.1]{AtkChaZoh} for more details). The main differences between the system \eqref{system for q_1 and q_2 tilde} for $\widetilde{q}_{1}$ and $\widetilde{q}_{2}$ with the system for $q_{1}$ and $q_{2}$ lie in the extra term $\frac{1}{\widetilde{q}_{1}^{3}}$ in the first equation of \eqref{system for q_1 and q_2 tilde}, as well as the small $\xi$ asymptotics of $\widetilde{q}_{1}(\xi)$, see \eqref{small x asymp q_1 tilde}.


\section{Model RH problem}
\label{sec: model rh problem}

In order to have compact notations in the coming sections, we define
\begin{equation}
\sigma_{3} = \begin{pmatrix}
1 & 0 \\ 0 & -1
\end{pmatrix}, \qquad N = \frac{1}{\sqrt{2}}\begin{pmatrix}
1 & i \\ i & 1
\end{pmatrix}.
\end{equation}
We also define for $y \in \mathbb{R}$ the following piecewise constant matrix:
\begin{equation}\label{def of H}
H_{y}(z) = \left\{  \begin{array}{l l}

I, & \mbox{ for } -\frac{2\pi}{3}< \arg(z-y)< \frac{2\pi}{3},\\

\begin{pmatrix}
1 & 0 \\
-e^{\pi i \alpha} & 1 \\
\end{pmatrix}, & \mbox{ for } \frac{2\pi}{3}< \arg(z-y)< \pi, \\

\begin{pmatrix}
1 & 0 \\
e^{-\pi i \alpha} & 1 \\
\end{pmatrix}, & \mbox{ for } -\pi< \arg(z-y)< -\frac{2\pi}{3}, \\

\end{array} \right.
\end{equation}
where the principal branch is chosen for the argument, such that $\arg(z-y) = 0$ for $z>y$.

\vspace{0.2cm}\hspace{-0.7cm} Let $0=x_{0}<x_{1}<...<x_{k}<+\infty$ and $s_{1},...,s_{k} \in [0,1]$, $s_{k+1} = 1$ be such that $s_{j+1} \neq s_{j}$ for $j \in \{1,...,k\}$. The solution of our model RH problem will be denoted by $\Phi(z;\vec{x},\vec{s})$, where $\vec{x} = (x_{1},...,x_{k})$ and $\vec{s} = (s_{1},...,s_{k})$. When there is no confusion,  we will just denote it by $\Phi(z)$ where the dependence in $\vec{x}$ and $\vec{s}$ is omitted.

\begin{figure}[t]
    \begin{center}
    \setlength{\unitlength}{1truemm}
    \begin{picture}(100,55)(20,10)
        \put(50,40){\line(1,0){60}}
        \put(50,40){\line(-1,0){30}}
        \put(50,40){\thicklines\circle*{1.2}}
        \put(72,40){\thicklines\circle*{1.2}}
        \put(87,40){\thicklines\circle*{1.2}}
        \put(100,40){\thicklines\vector(1,0){.0001}}
        \put(83,40){\thicklines\vector(1,0){.0001}}
        \put(110,40){\thicklines\circle*{1.2}}
        \put(50,40){\line(-0.5,0.866){15}}
        \put(50,40){\line(-0.5,-0.866){15}}
        \qbezier(53,40)(52,43)(48.5,42.598)
        \put(53,43){$\frac{2\pi}{3}$}
        \put(50.3,36.8){$-x_{3}$}
        \put(69,36.8){$-x_{2}$}
        \put(84.2,36.8){$-x_{1}$}
        \put(110.3,36.8){$0$}
        \put(40,60){$\Sigma_{1}$}
        \put(40,18){$\Sigma_{2}$}
        \put(62,52){$\mathcal{I}_{1}$}
        \put(27,52){$\mathcal{I}_{2}$}
        \put(27,30){$\mathcal{I}_{3}$}
        \put(62,30){$\mathcal{I}_{4}$}
        \put(65,39.9){\thicklines\vector(1,0){.0001}}
        \put(35,39.9){\thicklines\vector(1,0){.0001}}
        \put(41,55.588){\thicklines\vector(0.5,-0.866){.0001}}
        \put(41,24.412){\thicklines\vector(0.5,0.866){.0001}}
    \end{picture}
    \caption{\label{figPhi}The jump contour for $\Phi$ with $k=3$, and the four sectors $\mathcal{I}_{i}$, $i = 1,2,3,4$.}
\end{center}
\end{figure}
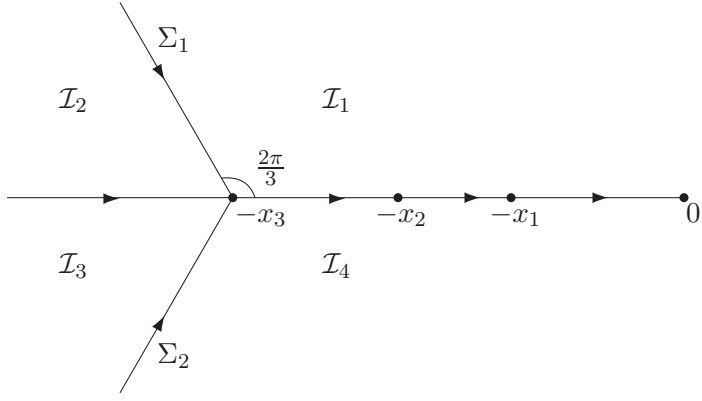

\subsubsection*{RH problem for $\Phi$}
\begin{itemize}
\item[(a)] $\Phi : \mathbb{C}\setminus \Sigma_{\Phi} \to \mathbb{C}^{2\times 2}$ is analytic, where the contour $\Sigma_{\Phi} = ((-\infty,0]\cup \Sigma_{1} \cup \Sigma_{2})$ is oriented as shown in Figure \ref{figPhi} with
\begin{equation*}
\Sigma_{1} = -x_{k}+ e^{\frac{2\pi i}{3}}\mathbb{R}^{+}, \qquad \Sigma_{2} = -x_{k}+ e^{-\frac{2\pi i}{3}}\mathbb{R}^{+}.
\end{equation*}
\item[(b)] The limits of $\Phi(z)$ as $z$ approaches $\Sigma_{\Phi}\setminus \{0,-x_{1},...,-x_{k}\}$ from the left ($+$ side) and from the right ($-$ side) exist, are continuous on $\Sigma_{\Phi}\setminus \{0,-x_{1},...,-x_{k}\}$ and are denoted by $\Phi_+$ and $\Phi_-$ respectively. Furthermore they are related by:
\begin{align}
& \Phi_{+}(z) = \Phi_{-}(z) \begin{pmatrix}
1 & 0 \\ e^{\pi i \alpha} & 1
\end{pmatrix}, & & z \in \Sigma_{1}, \\
& \Phi_{+}(z) = \Phi_{-}(z) \begin{pmatrix}
0 & 1 \\ -1 & 0
\end{pmatrix}, & & z \in (-\infty,-x_{k}), \\
& \Phi_{+}(z) = \Phi_{-}(z) \begin{pmatrix}
1 & 0 \\ e^{-\pi i \alpha} & 1
\end{pmatrix}, & & z \in \Sigma_{2}, \\
& \Phi_{+}(z) = \Phi_{-}(z) \begin{pmatrix}
e^{\pi i \alpha} & s_{j} \\ 0 & e^{-\pi i \alpha}
\end{pmatrix}, & & z \in (-x_{j},-x_{j-1}),
\end{align}
where $j =1,...,k$.
\item[(c)] As $z \to \infty$, we have 
\begin{equation}\label{Phi inf}
\Phi(z) = \Big(I + \Phi_{1}(\vec{x},\vec{s})z^{-1}+\bigO(z^{-2})\Big)z^{-\frac{\sigma_{3}}{4}}Ne^{z^{\frac{1}{2}}\sigma_{3}},
\end{equation}
where the principal branch is chosen for each root, and $\Phi_{1}$ is given by
\begin{equation}\label{def of Phi_1}
\Phi_{1}(\vec{x},\vec{s}) = \begin{pmatrix}
v(\vec{x},\vec{s}) & -it(\vec{x},\vec{s}) \\ ip(\vec{x},\vec{s}) & -v(\vec{x},\vec{s})
\end{pmatrix}.
\end{equation}
The fact that $\Phi_{1}$ is traceless follows directly from the relation $\det \Phi \equiv 1$.

As $z$ tends to $-x_{j}$, $j \in \{1,...,k\}$, $\Phi$ takes the form
\begin{equation}\label{Phi x_j}
\Phi(z) = \Phi_{0,j}(z) \begin{pmatrix}
1 & \frac{s_{j+1}-s_{j}}{2\pi i} \log (z+x_{j}) \\ 0 & 1
\end{pmatrix} V_{j}(z) e^{\frac{\pi i \alpha}{2}\theta(z)\sigma_{3}}H_{-x_{k}}(z),
\end{equation}
where $\Phi_{0,j}(z) = \Phi_{0,j}(z;\vec{r},\vec{s})$ is analytic in a neighbourhood of $(-x_{j+1},-x_{j-1})$, satisfies $\det \Phi_{0,j} \equiv 1$, and $\theta(z)$, $V_{j}(z)$ are piecewise constant and defined by
\begin{align}
& \theta(z) = \left\{  \begin{array}{l l}
+1, & \mbox{Im } z > 0, \\
-1, & \mbox{Im } z < 0,
\end{array} \right. & & V_{j}(z) = \left\{  \begin{array}{l l}
I, & \mbox{Im } z > 0, \\
\begin{pmatrix}
1 & -s_{j} \\ 0 & 1
\end{pmatrix}, & \mbox{Im } z < 0.
\end{array} \right.
\end{align}

As $z$ tends to $0$, the behaviour of $\Phi$ is
\begin{equation}\label{Phi 0}
\Phi(z) = \Phi_{0,0}(z)z^{\frac{\alpha}{2}\sigma_{3}} \begin{pmatrix}
1 & s_{1}h(z) \\ 0 & 1
\end{pmatrix}, \qquad \alpha > -1,
\end{equation}
where $\Phi_{0,0}(z)$ is analytic in a neighbourhood of $(-x_{1},\infty)$, satisfies $\det \Phi_{0,0} \equiv 1$ and 
\begin{equation}\label{def of h}
h(z) = \left\{ \begin{array}{l l}
\displaystyle \frac{1}{2 i \sin(\pi \alpha)}, & \alpha \notin \mathbb{N}, \\[0.35cm]
\displaystyle \frac{(-1)^{\alpha}}{2\pi i} \log z, & \alpha \in \mathbb{N}.
\end{array} \right.
\end{equation}
\end{itemize}
\begin{remark}\label{remark: uniqueness}
The solution of the RH problem for $\Phi$ is unique. This follows by standard arguments, based on the fact that $\det \Phi(z) \equiv 1$, see e.g. \cite[Theorem 7.18]{Deift}. We will prove the existence of the solution in Section \ref{Section: proof of main thm}, see in particular \eqref{Phi in terms of Y} and comments below.
\end{remark}
\begin{remark}
We can verify that $\sigma_{3} \overline{\Phi(\overline{z})}\sigma_{3}$ is also a solution of the RH problem for $\Phi$. Thus, by uniqueness of the solution (see Remark \ref{remark: uniqueness}), we have
\begin{equation}\label{symmetry of the RH problem}
\Phi(z) = \sigma_{3} \overline{\Phi(\overline{z})}\sigma_{3}.
\end{equation}
This means that there is some symmetry in the problem. In particular, this relation implies that the functions $v$, $t$ and $p$ that appear in \eqref{def of Phi_1} are real.
\end{remark}

\subsection*{Lax pair}
In this subsection, we obtain a system of $k$ ordinary differential equations for $k$ functions associated to $\Phi$. We derive these equations using Lax pair techniques. The following computations are similar to those done in \cite{AtkChaZoh} for the distribution of the ratio between the two smallest eigenvalues in the Laguerre Unitary Ensemble. We introduce a new parameter $x>0$, and we begin with the following transformation on $\Phi$:
\begin{equation}\label{def of Phi tilde}
\widetilde \Phi(z;x) = \widetilde E(x) \Phi(x^{2}z;\vec{r}x^{2},\vec{s}), \qquad \widetilde E(x) = \begin{pmatrix}
1 & 0 \\ \frac{t(\vec{r}x^{2},\vec{s})}{x} & 1
\end{pmatrix}e^{\frac{\pi i}{4}\sigma_{3}}x^{\frac{\sigma_{3}}{2}},
\end{equation}
where we have omitted the dependence of $\widetilde \Phi$ in $\vec{r}$ and $\vec{s}$. Note that with this transformation, $\widetilde{\Phi}$ satisfies an RH problem whose contour does not depend on $x$. By standard arguments, $\widetilde \Phi(z;x)$ is analytic in $x$ for $x$ in a compact subset of $(0,\infty)$. By differentiating $\widetilde \Phi$ with respect to $z$ and $x$, we obtain a Lax pair of the form
\begin{equation}\label{Lax pair}
\left\{ \begin{array}{l}
\displaystyle \partial_{z} \widetilde \Phi(z;x) = A(z;x)\widetilde \Phi(z;x), \\
\displaystyle \partial_{x} \widetilde \Phi(z;x) = B(z;x)\widetilde \Phi(z;x),
\end{array} \right.
\end{equation}
where we have also omitted the dependence of $A$ and $B$ in $\vec{r}$ and $\vec{s}$. Since $\widetilde \Phi$, $\partial_{z} \widetilde \Phi(z;x)$ and $\partial_{x} \widetilde \Phi(z;x)$ have the same jumps, $A$ and $B$ are meromorphic in $z \in \mathbb{C}$. From \eqref{Phi x_j} and \eqref{Phi 0}, $B$ is an entire function in $z$ and $A$ has simple poles in $z$ at $0$, $-r_{1}$,...,$-r_{k}$. We can use \eqref{Phi inf} to obtain an explicit expression for $B$:
\begin{equation}
B(z;x) = B_{0}(x) + z B_{1}, \qquad B_{0}(x) =  \begin{pmatrix}
0 & 1 \\ u(x) & 0
\end{pmatrix}, \quad B_{1} = \begin{pmatrix}
0 & 0 \\ 1 & 0
\end{pmatrix}.
\end{equation}
where $u(x) = \frac{2t^{\prime}(\vec{r}x^{2},\vec{s})x^{2}+t(\vec{r}x^{2},\vec{s})^{2}-2v(\vec{r}x^{2},\vec{s})-t(\vec{r}x^{2},\vec{s})}{x^{2}}$, and $t^{\prime}(\vec{r}x^{2}) = \left.\partial_{y}t(\vec{r}y)\right|_{y=x^{2}}$. On the other hand, $A$ can be written as
\begin{equation}\label{meromorphic A}
A(z;x) = A_{\infty}(x) + \sum_{j=0}^{k} \frac{A_{j}(x)}{z+r_{j}}.
\end{equation}
The matrix $A_{\infty}$ can also be explicitly evaluated by using \eqref{Phi inf}, we have
\begin{equation}\label{A inf}
A_{\infty}(x) = \begin{pmatrix}
0 & 0 \\ \frac{x}{2} & 0
\end{pmatrix}.
\end{equation}
Since $\det \widetilde \Phi (z)$ is constant, $A$ is traceless and we can also write
\begin{equation}\label{def of b_j}
A(z;x) = \begin{pmatrix}
a(z;x) & b(z;x) \\ c(z;x) & -a(z;x)
\end{pmatrix}, \qquad b(z;x) = \sum_{j=0}^{k}\frac{b_{j}(x)}{z+r_{j}}.
\end{equation}
We will derive a system of ordinary differential equations for $b_{0}(x)$, $b_{1}(x),..., b_{k}(x)$ and $u(x)$ from the compatibility condition 
\begin{equation}
\partial_{x}\partial_{z} \widetilde \Phi(z;x) = \partial_{z}\partial_{x} \widetilde \Phi(z;x),
\end{equation}
which by using \eqref{Lax pair} is equivalent to 
\begin{equation}
\partial_{x} A - \partial_{z}B + AB-BA = 0.
\end{equation}
This condition gives rise to the three following equations for $a$, $b$, $c$, and $u$:
\begin{align}
& 0 = c - b(z+u)-a^{\prime}, \label{aa6}\\
& 0 = 2a +b^{\prime}, \label{aa7}\\
& 0 = 2a(z+u) -c^{\prime} + 1,
\end{align}
where primes denote derivatives with respect to $x$. In particular $a$ and $c$ can be expressed in terms of $b$. Thus we can write the determinant of $A$ as
\begin{equation}\label{det A in terms of b}
\det A = -b^{2} (z+u) + \frac{(b^{2})^{\prime\prime}}{4} - \frac{3}{4} (b^{\prime})^{2}.
\end{equation}
Expanding $\det A (z)$ around $z = 0$, $-r_{1}$,...,$-r_{k}$ and $\infty$ using on one hand \eqref{def of b_j} and \eqref{det A in terms of b}, and on the other hand \eqref{Phi inf}, \eqref{Phi x_j} and \eqref{Phi 0}, and by expanding $A_{12}(z) = b(z)$ around $z = \infty$, we obtain
\begin{align}
& \sum_{j=0}^{k} b_{j}(x) = \frac{x}{2}, \label{aa9}\\
& (u(x)-r_{j})b_{j}(x)^{2} +\frac{1}{4}b_{j}^{\prime}(x)^{2} - \frac{1}{2}b_{j}(x)b_{j}^{\prime\prime}(x) = 0, \qquad j = 1,...,k \label{aa10}\\
& u(x)b_{0}(x)^{2} +\frac{1}{4}b_{0}^{\prime}(x)^{2} - \frac{1}{2}b_{0}(x)b_{0}^{\prime\prime}(x) = \frac{\alpha^{2}}{4}. \label{aa11}
\end{align}
\begin{definition}
We define $q_{j}$ in terms of $b_{j}$ as follows:
\begin{equation}\label{def of q_j}
q_{j}^{2}(x) = \frac{2b_{j}(\sqrt{x})}{\sqrt{x}}, \qquad j = 1,...,k.
\end{equation}
We can use \eqref{aa9} and \eqref{aa11} to express $u$ and $b_{0}$ in terms of $b_{1}$,...,$b_{k}$, and therefore in terms of $q_{1}$,...,$q_{k}$. By substituting these expressions for $u$ and $b_{0}$ in \eqref{aa10}, we obtain the system of $k$ coupled Painlev\'{e} V equations given by \eqref{system of ODEs for the q_j}. Also, from \eqref{symmetry of the RH problem}, if $z \in \mathbb{R}\setminus \{-r_{k},...,-r_{1},0\}$, we have $b(z;x) = \overline{b(z;x)}$. This implies that $b_{0}$,...,$b_{k}$, and therefore $q_{1}^{2}$,...,$q_{k}^{2}$, are all real functions of $x \in \mathbb{R}^{+}$.
\end{definition}
Proposition \ref{Prop: Phi in terms of q_j} below will be useful in Section \ref{Section: proof of main thm} to integrate the identity \eqref{diff identity main identity}.
\begin{proposition}\label{Prop: Phi in terms of q_j} For each $j = 1,2,\dots,k$,  there holds the relation
\begin{equation}
\partial_{x} \left( x \lim_{z \to -r_{j}x}[\Phi^{-1}(z;\vec{r}x,\vec{s})\Phi^{\prime}(z;\vec{r}x,\vec{s})]_{21} \right) =  \frac{2\pi i e^{-\pi i \alpha}}{s_{j+1}-s_{j}} \frac{q_{j}^{2}(x)}{4},
\end{equation}
where the limit is taken from $z \in \mathcal{I}_{4}$, with $\mathcal{I}_{4}$ as shown in Figure \ref{figPhi}, and where $\Phi^{\prime} = \partial_{z}\Phi$.
\end{proposition}
\begin{proof}
We recall that $\Phi_{0,j}(z;\vec{r}x,\vec{s})$ defined in \eqref{Phi x_j} is invertible and analytic in $z$ in a neighbourhood of $-r_{j}x$. 
By expanding it around $-r_{j}x$, we can write
\begin{equation}
\Phi_{0,j}(z;\vec{r}x,\vec{s}) = E_{j}(x)\big(I+F_{j}(x)(z+r_{j}x) + \bigO((z+r_{j}x)^{2})\big), \qquad \mbox{ as } z \to -r_{j}x,
\end{equation}
for certain matrices $E_{j}$ and $F_{j}$ (they depend also on $\vec{r}$ and $\vec{s}$). Therefore, we have
\begin{equation}\label{aa2}
\begin{array}{r c l}
\displaystyle \lim_{z \to -r_{j}x}[\Phi^{-1}(z;\vec{r}x,\vec{s})\Phi^{\prime}(z;\vec{r}x,\vec{s})]_{21} & = & \displaystyle e^{-\pi i \alpha}[\Phi_{0,j}^{-1}(-r_{j}x)\Phi_{0,j}^{\prime}(-r_{j}x)]_{21}, \\
& = & \displaystyle e^{-\pi i \alpha} F_{j,21}(x),
\end{array}
\end{equation}
where the limit is taken from $z \in \mathcal{I}_{4}$.
On the other hand, taking the limit $z\to -r_{j}$ in the $B$-equation in the Lax pair \eqref{Lax pair} leads to
\begin{align}
& \partial_{x} (\widetilde E(x) E_{j}(x^{2})) = (B_{0}(x)-r_{j}B_{1})\widetilde E(x) E_{j}(x^{2}) - \widetilde E(x)E_{j}(x^{2})K(x), \\
& \partial_{x} (x^{2}F_{j}(x^{2})) = (\widetilde E (x) E_{j}(x^{2}))^{-1}B_{1}\widetilde E(x) E_{j}(x^{2}) + x^{2} [K(x),F_{j}(x^{2})], \label{lol 1}
\end{align}
where $K(x) = \begin{pmatrix}
0 & \frac{s_{j+1}-s_{j}}{\pi i x}  \\ 0 & 0
\end{pmatrix}$. In particular, taking the $(2,1)$ entry in \eqref{lol 1} and using the fact that $\det E_{j}(x) = 1$ leads to
\begin{equation}
\partial_{x} (x^{2}F_{j,21}(x^{2})) = i x E_{j,11}(x^{2})^{2}.
\end{equation}
By the change of variables $x^{2} \mapsto x$, this can be rewritten as
\begin{equation}\label{aa3}
\partial_{x} (x F_{j,21}(x)) = \frac{i}{2} E_{j,11}(x)^{2}.
\end{equation}
We also have, by the $A$-equation in the Lax pair \eqref{Lax pair}, as $z \to -r_{j}$
\begin{equation}\label{aa1}
A(z;x) = \frac{s_{j+1}-s_{j}}{2\pi i(z+r_{j})} \widetilde E(x) \begin{pmatrix}
-E_{j,11}(x^{2})E_{j,21}(x^{2}) & E_{j,11}(x^{2})^{2} \\
-E_{j,21}(x^{2})^{2} & E_{j,11}(x^{2})E_{j,21}(x^{2})
\end{pmatrix}\widetilde E(x)^{-1} + \bigO(1).
\end{equation}
Equation \eqref{aa1} implies then
\begin{equation}
b_{j}(x) = \frac{s_{j+1}-s_{j}}{2\pi }x E_{j,11}(x^{2})^{2}.
\end{equation}
Thus by \eqref{aa2}, \eqref{aa3} and \eqref{def of q_j}, we obtain the claim.
\end{proof}

\section{Proof of Theorem \ref{thm: main thm}}
\label{Section: proof of main thm}

We start the proof of Theorem \ref{thm: main thm} by following a theory developed by Its, Izergin, Korepin and Slavnov \cite{IIKS}, which was also developed by Bertola and Cafasso in \cite{BertolaCafasso}, to express the quantities $\partial_{x_{j}} \log F(\vec{x},\vec{s})$, $j=1,...,k$ in terms of a RH problem related to an \emph{integrable kernel} $R$ (the solution of this RH problem will be denoted $Y$). Then, we will relate $Y$ to $\Phi$ and finally integrate these identities. Let $K:\mathbb{R}^{+}\times \mathbb{R}^{+} \to \mathbb{R}$ be given by
\begin{equation}
K(u, v) = \chi_{(0,x_k)}(u) \sum_{j = 1}^k (1 - s_j) K^{\textrm{Be}}(u, v) \chi_{(x_{j-1}, x_j)}(v), \qquad u,v >0.
\end{equation}
This is the kernel of a trace class integral operator $\mathcal{K}$ acting on $L^2(\R^+)$. The kernel $K$ is integrable in the sense of Its, Izergin, Korepin and Slavnov, i.e. it can be written in the form
\begin{equation}
K(u,v) = \frac{f^{T}(u)g(v)}{u-v}, \qquad f^{T}(u)g(u) = 0, \qquad u,v >0,
\end{equation}
where $f(u)$ and $g(v)$ are given by
\begin{equation*}
f(u) = \frac{1}{2}\begin{pmatrix}
\chi_{(0,x_{k})}(u) J_{\alpha}(\sqrt{u}) \\
\chi_{(0,x_{k})}(u) \sqrt{u} J_{\alpha}^{\prime}(\sqrt{u})
\end{pmatrix}, \quad g(v) = \begin{pmatrix}
\sum_{j=1}^{k}(1-s_{j}) \sqrt{v} J_{\alpha}^{\prime}(\sqrt{v})\chi_{(x_{j-1},x_{j})}(v) \\
-\sum_{j=1}^{k}(1-s_{j}) J_{\alpha}(\sqrt{v})\chi_{(x_{j-1},x_{j})}(v)
\end{pmatrix}.
\end{equation*}
Also, by using the connection formula $I_{\alpha}(e^{\frac{\pi i}{2}}\sqrt{u})=e^{\frac{\alpha \pi i}{2}}J_{\alpha}(\sqrt{u})$ for $u >0$ (see \cite[formula 10.27.6]{NIST}), $f(u)$ and $g(v)$ can be rewritten in terms of $\widetilde{P}_{\mathrm{Be}}$ (this is the solution of a modified Bessel model RH problem, and is defined in the Appendix, see \eqref{def of P_Be tilde}) as follows:
\begin{align}
& f(u) = \frac{e^{-\frac{\alpha \pi i}{2}}e^{\frac{\pi i}{4}}}{2\sqrt{\pi}} \chi_{(0,x_{k})}(u) \sigma_{3} \widetilde{P}_{\mathrm{Be},+}(-u) \begin{pmatrix}
1 \\ 0
\end{pmatrix}, & & \mbox{ for } u >0, \label{lol 10}\\
& g(v) = \frac{e^{-\frac{\alpha \pi i}{2}}e^{\frac{\pi i}{4}}}{\sqrt{\pi}}\sum_{j=1}^{k}(1-s_{j})\chi_{(x_{j-1},x_{j})}(v) \sigma_{3} \widetilde{P}_{\mathrm{Be},+}^{-1}(-v)^{T}\begin{pmatrix}
0 \\ 1
\end{pmatrix} & & \mbox{ for } v >0. \label{lol 11}
\end{align}
In the Bessel point process, for all bounded Borel set $B$ with non zero Lebesgue measure, we have $\mathbb{P}(n_B = 0) > 0$. Therefore, from \eqref{def of F} and \eqref{F Fredholm} we have $\det(1 - \mathcal{K}) > 0$ if $s_1,...,s_k \in [0,1]$. By standard properties of trace class operators (see e.g. \cite[page 1029]{Schwartz}), we have
\begin{equation}\label{lol 7}
\partial_{x_{j}} \log \det (1-\mathcal{K}) = - \mbox{Tr} \left( (1-\mathcal{K})^{-1}\partial_{x_{j}}\mathcal{K} \right), \qquad j =1,...,k.
\end{equation}
In our case, it can be rewritten more explicitly as
\begin{equation}\label{lol 8}
\begin{array}{r c l}
\mbox{Tr} \left( (1-\mathcal{K})^{-1}\partial_{x_{j}}\mathcal{K} \right) & = & \displaystyle (s_{j+1}-s_{j}) \mbox{Tr}\left( (1-\mathcal{K})^{-1}\mathcal{K}^{\mathrm{Be}}\delta_{x_{j}} \right) \\[0.1cm]
& = & \displaystyle \frac{s_{j+1}-s_{j}}{1-s_{j}} \lim_{u \nearrow x_{j}} [(1-K)^{-1}K](u,u) \\[0.1cm]
& = & \displaystyle \frac{s_{j+1}-s_{j}}{1-s_{j}} \lim_{u \nearrow x_{j}} R(u,u)
\end{array}
\end{equation}
where $R$ is the kernel for the resolvent operator $\mathcal{R}$ defined by
\begin{equation}
1 + \mathcal{R} = (1 - \mathcal{K})^{-1}.
\end{equation}
If $s_{j} = 1$, then we take the limit $z\searrow x_{j}$ instead, and the above formula is replaced by
\begin{equation}
\mbox{Tr} \left( (1-\mathcal{K})^{-1}\partial_{x_{j}}\mathcal{K} \right) = \frac{s_{j+1}-s_{j}}{1-s_{j+1}} \lim_{z\searrow x_{j}} R(z,z),
\end{equation}
which is well defined since $s_{j+1} \neq s_j$. Let us now define the matrix $Y$ by
\begin{equation}\label{def of Y}
Y(z) = I - \int_0^{x_k} \frac{F(\mu)g^T(\mu)}{\mu - z} d\mu, \qquad F(\mu) = \begin{pmatrix} (1-\mathcal{K})^{-1} f_1(\mu) \\ (1 - \mathcal{K})^{-1} f_2(\mu) \end{pmatrix}.
\end{equation}
The function $Y$ satisfies the following RH problem \cite{DeiftItsZhou}.
\subsubsection*{RH problem for Y}
\begin{itemize}
\item[(a)] $Y : \mathbb{C}\setminus [0,x_{k}] \to \mathbb{C}^{2\times 2}$ is analytic
\item[(b)] For $u \in (0,x_{k}) \setminus \{x_1,...,x_k\}$, the limits $\lim_{\epsilon \to 0_{+}} Y(u\pm i \epsilon)$ exist, are denoted $Y_{+}(u)$ and $Y_{-}(u)$ respectively, are continuous as functions of $u \in (0,x_{k})$, and satisfy furthermore the jump relation
\begin{equation}
Y_{+}(u) = Y_{-}(u)J_{Y}(u), \qquad J_{Y}(u) = I - 2\pi i f(u)g^{T}(u).
\end{equation}
\item[(c)] $Y(z) = I + \bigO(z^{-1})$ as $z \to \infty$.
\item[(d)] $Y(z) = \bigO(\log(z-x_j))$ as $z \to x_{j}$, for each $j = 0,...,k$ (with $x_0 = 0$).
\end{itemize}
For $u,v \in (0,x_{k})$, the resolvent can now be written as \cite{DeiftItsZhou}
\begin{equation}\label{R in terms of Y}
\hspace{-0.8cm} R(u,v) = \frac{F^{T}(u)G(v)}{u-v}, \hspace{-0.1cm} \quad \mbox{where} \hspace{-0.1cm} \quad F(u) = Y_{+}(u)f(u) \hspace{-0.1cm} \quad \mbox{and} \hspace{-0.1cm} \quad G(v) = (Y_{+}^{-1}(v))^{T}g(v).
\end{equation}
Now we want to relate $Y$ with $\Phi$. Let us consider $X(z) = \widetilde{Y}(z) \widetilde{P}_{\mathrm{Be}}(z)$, where $\widetilde{Y}(z) = \sigma_{3}Y(-z)\sigma_{3}$ and $\widetilde{P}_{\mathrm{Be}}$ is the solution of a modified Bessel model RH problem, defined in \eqref{def of P_Be tilde}. Since $\widetilde{Y}$ is analytic on $\Sigma_{1} \cup \Sigma_{2} \cup (-\infty,-x_{k})$, from the jumps of $\widetilde{P}_{\mathrm{Be}}$, it is direct that $X$ has exactly the same jumps as $\Phi$ on $\Sigma_{1} \cup \Sigma_{2} \cup (-\infty,-x_{k})$. The jumps $J_{X}$ of $X$ are \emph{a priori} more involved on $(-x_{k},0)$. They are given by
\begin{equation}\label{lol 12}
J_{X}(-u) = \begin{pmatrix}
e^{\pi i  \alpha} & 1 \\ 0 & e^{-\pi i \alpha}
\end{pmatrix}\widetilde{P}_{\mathrm{Be},+}^{-1}(-u)J_{\widetilde{Y}}(-u)\widetilde{P}_{\mathrm{Be},+}(-u), \qquad u \in (0,x_{k}),
\end{equation}
where $J_{\widetilde{Y}}$ is the jump of $\widetilde{Y}$, given by
\begin{equation}
J_{\widetilde{Y}}(-u) = \sigma_{3} J_{Y}(u)^{-1} \sigma_{3}, \qquad u \in (0,x_{k}).
\end{equation}
For $u \in (0,x_{k})$, by \eqref{lol 10} and \eqref{lol 11}, we have
\begin{equation}
J_{\widetilde{Y}}(-u) = \widetilde{P}_{\mathrm{Be},+}(-u)\begin{pmatrix}
1 & -e^{-\pi i \alpha}\sum_{j=1}^{k}(1-s_{j})\chi_{(x_{j-1},x_{j})}(u) \\
0 & 1
\end{pmatrix}\widetilde{P}_{\mathrm{Be},+}^{-1}(-u).
\end{equation}
By plugging it into \eqref{lol 12}, $J_{X}$ is simply reduced to
\begin{equation}
J_{X}(-u) = \begin{pmatrix}
e^{\pi i \alpha} & \sum_{j=1}^{k}s_{j}\chi_{(x_{j-1},x_{j})}(u) \\ 0 & e^{-\pi i \alpha}
\end{pmatrix}, \qquad u \in (0,x_{k}),
\end{equation}
which is precisely the same jump as $\Phi(z;\vec{x},\vec{s})$ for $z \in (-x_{k},0)$. On the other hand, from \eqref{large z asymptotics modified Bessel}, as $z \to \infty$ we have
\begin{equation}
X(z) = e^{-\frac{\pi i}{4}\sigma_{3}}
\begin{pmatrix}
1 & 0 \\ \frac{-i}{8}(4\alpha^{2}+3) & 1
\end{pmatrix}\left(I+\bigO( z ^{-1})\right)  z ^{\frac{-\sigma_{3}}{4}} N e^{ z ^{\frac{1}{2}}\sigma_{3}}.
\end{equation}
Thus by uniqueness of the solution of the RH problem for $\Phi$, see Remark \ref{remark: uniqueness}, we have
\begin{equation}\label{Phi in terms of Y}
\Phi(z;\vec{x},\vec{s}) = \begin{pmatrix}
1 & 0 \\ \frac{i}{8}(4\alpha^{2}+3) & 1
\end{pmatrix} e^{\frac{\pi i }{4}\sigma_{3}} \widetilde{Y}(z) \widetilde{P}_{\mathrm{Be}}(z).
\end{equation}
Since from our proof, the matrix $\widetilde{Y}$ on the right hand side exists and is constructed explicitly in terms of $(1-\mathcal{K})^{-1}$ (see \eqref{def of Y}), it also proves the existence of a solution for the RH problem for $\Phi$. Note that \eqref{lol 10} and \eqref{lol 11} can equivalently be written as
\begin{equation*}
\hspace{-0.4cm}\widetilde{P}_{\mathrm{Be},-}^{-1}(-u)\sigma_{3}f(u) = \frac{c}{2}  \chi_{(0,x_{k})}(u) \begin{pmatrix}
1 \\ 0
\end{pmatrix}, \hspace{-0.25cm} \quad \widetilde{P}_{\mathrm{Be},-}(-v)^{T}\hspace{-0.1cm}\sigma_{3}g(v) = c\sum_{j=1}^{k}(1-s_{j})\chi_{(x_{j-1},x_{j})}(v)  \begin{pmatrix}
0 \\ 1
\end{pmatrix}\hspace{-0.1cm},
\end{equation*}
where $u,v \in \mathbb{R}^{+}$ and $c = \frac{e^{\frac{\pi i}{4}}e^{\frac{\alpha \pi i}{2}}}{\sqrt{\pi}}$. Thus for $u,v \in \mathbb{R}^{+}$, we have
\begin{equation}
R(u,v) = \frac{c^{2}}{2} \frac{[\Phi_{-}^{-1}(-v;\vec{x},\vec{s})\Phi_{-}(-u;\vec{x},\vec{s})]_{21}}{u-v}\chi_{(0,x_{k})}(u) \sum_{j=1}^{k}(1-s_{j})\chi_{(x_{j-1},x_{j})}(v) .
\end{equation}
By taking the limit $v\to u$ for a certain $u \in (x_{j-1},x_{j})$ with $j \in \{1,...,k\}$ in the above expression, we obtain
\begin{equation}\label{lol 17}
R(u,u) = - \frac{c^{2}}{2}(1-s_{j}) \left[ \Phi_{-}(-u;\vec{x},\vec{s})^{-1} \Phi_{-}^{\prime}(-u;\vec{x},\vec{s}) \right]_{21}.
\end{equation}
Taking now the limit $u \nearrow x_{j}$ in \eqref{lol 17} and substituting the result in \eqref{lol 7} and \eqref{lol 8}, we obtain an explicit differential identity in terms of $\Phi$ for each $j \in \{1,...,k\}$:
\begin{equation}\label{k diff identities}
\partial_{x_{j}} \log F(\vec{x},\vec{s}) = -(s_{j+1}-s_{j}) \frac{e^{\pi i \alpha}}{2\pi i} \lim_{z \to -x_{j}}[\Phi^{-1}(z;\vec{x},\vec{s})\Phi^{\prime}(z;\vec{x},\vec{s})]_{21},
\end{equation}
where the limit is taken from $z \in \mathcal{I}_{4}$, with $\mathcal{I}_{4}$ as shown in Figure \ref{figPhi}. By simple compositions, we can use the above identities to get
\begin{equation}\label{diff identity main identity}
\begin{array}{r c l}
\hspace{-0.3cm}\displaystyle \partial_{x} \log F(\vec{r}x,\vec{s}) & = & \displaystyle \sum_{j=1}^{k} r_{j} \partial_{x_{j}} \log \left.F(\vec{x},\vec{s})\right|_{\vec{x}= \vec{r}x} \\
& = & \displaystyle -\sum_{j=1}^{k} r_{j} (s_{j+1}-s_{j}) \frac{e^{\pi i \alpha}}{2\pi i}\lim_{z \to -r_{j}x}[\Phi^{-1}(z;\vec{r}x,\vec{s})\Phi^{\prime}(z;\vec{r}x,\vec{s})]_{21}.
\end{array}
\end{equation}
Let $\epsilon$ and $x$ be such that $0<\epsilon < x$. By integrating the above expression from $\epsilon$ to $x$, this gives
\begin{equation}\label{lol 18}
\hspace{-0.9cm}\log \frac{F(\vec{r}x,\vec{s})}{F(\vec{r}\epsilon,\vec{s})} = -\sum_{j=1}^{k} r_{j} (s_{j+1}-s_{j}) \frac{e^{\pi i \alpha}}{2\pi i} \int_{\epsilon}^{x} \lim_{z \to -r_{j}\xi}[\Phi^{-1}(z;\vec{r}\xi,\vec{s})\Phi^{\prime}(z;\vec{r}\xi,\vec{s})]_{21}d\xi.
\end{equation}
Integrating it by parts and using Proposition \ref{Prop: Phi in terms of q_j}, one has
\begin{multline}\label{lol 9}
\hspace{-1cm}\int_{\epsilon}^{x} \lim_{z \to -r_{j}\xi}[\Phi^{-1}(z;\vec{r}\xi,\vec{s})\Phi^{\prime}(z;\vec{r}\xi,\vec{s})]_{21}d\xi =  \log \left(\frac{x}{\epsilon}\right)\epsilon \lim_{z \to -r_{j}\epsilon}[\Phi^{-1}(z;\vec{r}\epsilon,\vec{s})\Phi^{\prime}(z;\vec{r}\epsilon,\vec{s})]_{21}  \\ +  \frac{2\pi ie^{-\pi i \alpha}}{s_{j+1}-s_{j}} \int_{\epsilon}^{x} \log \left( \frac{x}{\xi} \right) \frac{q_{j}^{2}(\xi)}{4}d\xi.
\end{multline}
We will prove in the next section that 
\begin{align}
& \lim_{\epsilon \to 0} \log \left(\frac{x}{\epsilon}\right)\epsilon \lim_{z \to -r_{j}\epsilon}[\Phi^{-1}(z;\vec{r}\epsilon,\vec{s})\Phi^{\prime}(z;\vec{r}\epsilon,\vec{s})]_{21}  = 0, \label{small eps limit phi} \\
& \mbox{ and that } \quad \int_{0}^{x} \log \left( \frac{\xi}{x} \right) q_{j}^{2}(\xi)d\xi \in \mathbb{R} \quad \mbox{ for every } j \in \{1,...,k\}. \label{existence of the integral}
\end{align}
Thus, taking the limit $\epsilon \to 0$ in \eqref{lol 9} and in \eqref{lol 18} gives, using $F(\vec{0},\vec{s}) = 1$ (see \eqref{F Fredholm}), the following identity
\begin{equation}
\log F(\vec{r}x,\vec{s}) = - \sum_{k=1}^{k} \frac{r_{j}}{4} \int_{0}^{x} \log \left( \frac{x}{\xi} \right) q_{j}^{2}(\xi)d\xi.
\end{equation}
Apart from \eqref{small eps limit phi} and \eqref{existence of the integral}, this finishes the proof of Theorem \ref{thm: main thm}.

\section{Small $x$ asymptotics}
\label{Section: small x asymptotics}

In this section, we perform a Deift/Zhou steepest descent \cite{DeiftZhou1992, DeiftZhou, DKMVZ2, DKMVZ1} to obtain small $x$ asymptotics for $\Phi(z;\vec{r}x,\vec{s})$ uniformly in $z$, and where $\vec{r}$ and $\vec{s}$ are independent of $x$ and satisfy conditions \eqref{condition r} and \eqref{condition s}.

\subsection{First transformation $\Phi \mapsto W$}

The first transformation consists of making the rays $\Sigma_{1}$ and $\Sigma_{2}$ ending at $0$ instead of $-r_{k}x$, we define
\begin{equation}
W(z) = \Phi(z;\vec{r}x,\vec{s}) H_{-r_{k}x}(z)^{-1}H_{0}(z).
\end{equation}
It satisfies the following RH problem.
\subsubsection*{RH problem for $W$}
\begin{itemize}
\item[(a)] $W : \mathbb{C}\setminus \Big((-\infty,0]\cup e^{\pm \frac{2\pi i}{3}}\mathbb{R}^{+}\Big) \to \mathbb{C}^{2\times 2}$ is analytic, where the rays $e^{\pm \frac{2\pi i}{3}}\mathbb{R}^{+}$ are oriented from $e^{\pm \frac{2\pi i}{3}}\infty$ to $0$.
\item[(b)] The jumps for $W$ are given by 
\begin{align}
& \hspace{-0.5cm}W_{+}(z) = W_{-}(z) \begin{pmatrix}
1 & 0 \\ e^{\pi i \alpha} & 1
\end{pmatrix}, & & z \in e^{\frac{2\pi i}{3}}\mathbb{R}^{+}, \\
& \hspace{-0.5cm}W_{+}(z) = W_{-}(z) \begin{pmatrix}
0 & 1 \\ -1 & 0
\end{pmatrix}, & & z \in (-\infty,-r_{k}x), \\
& \hspace{-0.5cm}W_{+}(z) = W_{-}(z) \begin{pmatrix}
1 & 0 \\ e^{-\pi i \alpha} & 1
\end{pmatrix}, & & z \in e^{\frac{- 2\pi i}{3}}\mathbb{R}^{+}, \\
& \hspace{-0.5cm}W_{+}(z) = W_{-}(z) \begin{pmatrix}
e^{\pi i \alpha}(1-s_{j}) & s_{j} \\ s_{j}-2 & e^{-\pi i \alpha}(1-s_{j})
\end{pmatrix}, & & z \in (-r_{j}x,-r_{j-1}x),
\end{align}
where $j =1,...,k$.
\item[(c)] As $z \to \infty$, we have 
\begin{equation}\label{W inf}
W(z) = \big(I + \bigO(z^{-1})\big)z^{-\frac{\sigma_{3}}{4}}Ne^{z^{\frac{1}{2}}\sigma_{3}}.
\end{equation}

As $z$ tends to $-r_{j}x$, $j \in \{1,...,k\}$, the behaviour of $W$ is
\begin{equation}\label{W x_j}
W(z) = \Phi_{0,j}(z) \begin{pmatrix}
1 & \frac{s_{j+1}-s_{j}}{2\pi i} \log (z+r_{j}x) \\ 0 & 1
\end{pmatrix} V_{j}(z) e^{\frac{\pi i \alpha}{2}\theta(z)\sigma_{3}} H_{0}(z).
\end{equation}

As $z$ tends to $0$, the behaviour of $W$ is
\begin{equation}\label{W 0}
W(z) = \Phi_{0,0}(z)z^{\frac{\alpha}{2}\sigma_{3}} \begin{pmatrix}
1 & s_{1}h(z) \\ 0 & 1
\end{pmatrix}H_{0}(z).
\end{equation}
\end{itemize}
\subsection{Global parametrix}
Ignoring a small neighbourhood of $0$, we are left with a Riemann-Hilbert problem which is independent of $x$. We denote the solution of this RH problem $P^{(\infty)}$. The jumps of $P^{(\infty)}$, as well as its asymptotic behaviour at $\infty$ \eqref{W inf}, are the same of those of the Bessel model RH problem of order $\alpha$ presented in Appendix \ref{Appendix: Bessel} (the solution of the Bessel model RH problem is denoted $P_{\mathrm{Be}}(z;\alpha)$). If we don't specify the behaviour of the global parametrix near $z = 0$, the solution is not unique and for example $P^{(\infty)}(z) = P_{\mathrm{Be}}(z;\alpha + 2 n)$ for any $n \in \mathbb{N}$ is a solution. In order to have later the matching condition with the local parametrix, see \eqref{matching condition}, we choose the global parametrix to be
\begin{equation}\label{def of P inf}
P^{(\infty)}(z) = P_{\mathrm{Be}}(z;\alpha).
\end{equation}
\subsection{Local parametrix}
Inside a fixed disk $D_{0}$ around $0$, we want the local parametrix $P$ to satisfy the following RH problem:
\subsubsection*{RH problem for $P$}
\begin{itemize}
\item[(a)] $P : D_{0}\setminus \Big((-\infty,0]\cup e^{\pm \frac{2\pi i}{3}}\mathbb{R}^{+}\Big) \to \mathbb{C}^{2\times 2}$ is analytic.
\item[(b)] For $z \in D_{0} \cap \Big((-\infty,0]\cup e^{\pm \frac{2\pi i}{3}}\mathbb{R}^{+}\Big)$, $P(z)$ has the same jumps as $W(z)$, i.e. we have $P_{-}^{-1}(z)P_{+}(z) = W_{-}^{-1}(z)W_{+}(z)$.
\item[(c)] As $x \to 0$, we have 
\begin{equation}\label{matching condition}
P(z) = \big(I + \bigO(x)\big)P^{(\infty)}(z),
\end{equation}
uniformly for $z \in \partial D_{0}$.

\item[(d)]As $z$ tends to $-r_{j}x$, $j \in \{0,1,...,k\}$, we have 
\begin{equation}\label{local behaviour of the local parametrix}
W(z)P^{-1}(z) = \bigO(1).
\end{equation}
\end{itemize}
It can be directly verified that the following matrix satisfies conditions (a), (b) and (d) of the above RH problem:
\begin{equation}
P(z) = P_{\mathrm{Be},0}(z;\alpha)\begin{pmatrix}
1 & f(z;x) \\ 0 & 1
\end{pmatrix}z^{\frac{\alpha}{2}\sigma_{3}}\begin{pmatrix}
1 & h(z) \\ 0 & 1
\end{pmatrix}H_{0}(z),
\end{equation}
where $P_{\mathrm{Be},0}(z;\alpha)$ is analytic in a neighbourhood of $0$ and defined in \eqref{local Bessel behaviour}, and where $f$ is given by
\begin{equation}\label{def of f}
f(z;x) = \frac{-1}{2\pi i} \sum_{j=1}^{k}(1-s_{j}) \int_{-r_{j}x}^{-r_{j-1}x} \frac{|s|^{\alpha}}{s-z}ds.
\end{equation}
From \eqref{local Bessel behaviour} and \eqref{def of P inf}, we have
\begin{equation}
\hspace{-0.7cm}P(z)P^{(\infty)}(z)^{-1} = P_{\mathrm{Be},0}(z;\alpha) \begin{pmatrix}
1 & f(z;x) \\ 0 & 1 
\end{pmatrix} P_{\mathrm{Be},0}(z;\alpha)^{-1} = I + \bigO(x),
\end{equation}
as $x \to 0$ uniformly for $z \in \partial D_{0}$, and the matching condition \eqref{matching condition} holds.

\subsection{Small norm RH problem}
We define
\begin{equation}
R(z) = \left\{ \begin{array}{l l}
W(z)P^{(\infty)}(z)^{-1}, & z \in \mathbb{C}\setminus D_{0}, \\
W(z)P(z)^{-1}, & z \in  D_{0}.
\end{array} \right.
\end{equation}
Since $P^{(\infty)}$ (resp. $P$) has the same jumps as $W$ on $\mathbb{C}\setminus D_{0}$ (resp. on $D_{0}$), $R$ is analytic on $\mathbb{C}\setminus (\partial D_{0}\cup \{0,-r_{1}x,...,-r_{k}x\}$. Furthermore, from \eqref{local behaviour of the local parametrix}, $R$ is bounded near $0,-r_{1}x,...,-r_{k}x$ and thus $0,-r_{1}x,...,-r_{k}x$ are removable singularities. It follows that $R$ is analytic on $\mathbb{C}\setminus \partial D_{0}$. Also, from \eqref{W inf}, \eqref{def of P inf} and \eqref{large z asymptotics Bessel}, we have that $R(z) = I + \bigO(z^{-1})$ as $z \to \infty$. Let us put the clockwise orientation on $\partial D_{0}$. The jumps $J_{R}(z) = R_{-}^{-1}(z)R_{+}(z)$ satisfy 
\begin{equation}
J_{R}(z) = P(z)P^{(\infty)}(z)^{-1} = I+\bigO(x), \qquad \mbox{ as } x \to 0 \mbox{ uniformly for } z \in \partial D_{0},
\end{equation}
where we have used \eqref{matching condition}.
 It follows from standard theory for small norm RH problems that $R$ exists for sufficiently small $x$ and satisfies
\begin{equation}\label{asymptotics for R small x}
R(z) = I + \bigO(x), \qquad R^{\prime}(z) = \bigO(x),
\end{equation}
uniformly for $z \in \mathbb{C}\setminus \partial D_{0}$. We are now in a position to compute the small $x$ asymptotics of $b_{0}(x)$,...,$b_{k}(x)$. Inverting the transformations $R \mapsto W \mapsto \Phi$, we obtain for $z \in D_{0}$ that
\begin{equation}\label{explicit expression for Phi in the small x}
\Phi(z;\vec{r}x,\vec{s}) = R(z)P_{\mathrm{Be},0}(z;\alpha)\begin{pmatrix}
1 & f(z;x) \\ 0 & 1
\end{pmatrix}z^{\frac{\alpha}{2}\sigma_{3}}\begin{pmatrix}
1 & h(z) \\ 0 & 1
\end{pmatrix} H_{-r_{k}x}(z).
\end{equation}
By \eqref{def of Phi tilde}, \eqref{Lax pair} and \eqref{def of b_j}, for any $j \in \{1,...,k\}$ we have
\begin{equation}\label{lol 3}
b_{j}(\sqrt{x}) = i\sqrt{x} \lim_{z\to -r_{j}}(z+r_{j}) \left[ \partial_{z} (\Phi(xz;\vec{r}x,\vec{s}))\Phi^{-1}(xz;\vec{r}x,\vec{s}) \right]_{12}.
\end{equation}
Using \eqref{explicit expression for Phi in the small x} and the small $x$ asymptotics for $R$ given by \eqref{asymptotics for R small x}, after some calculations we obtain for $j \in \{1,...,k\}$ that
\begin{equation}
b_{j}(\sqrt{x}) = i \sqrt{x} (1+\bigO(x)) P_{\mathrm{Be},0,11}^{2}(0;\alpha) \lim_{z \to -r_{j}} (z+r_{j}) \partial_{z} \big(f(xz;x)\big).
\end{equation}
For $j \in \{1,...,k-1\}$, only two terms in the sum \eqref{def of f} contribute to this limit. After a straightforward calculation we have that
\begin{equation*}
\lim_{z\to -r_{j}} (z+r_{j})\hspace{-0.05cm} \int_{-r_{j}x}^{-r_{j-1}x} \hspace{-0.25cm}\frac{x|s|^{\alpha}}{(s-xz)^{2}}ds = - (r_{j}x)^{\alpha}, 
\end{equation*}
\begin{equation*}
\lim_{z\to -r_{j}} (z+r_{j})\hspace{-0.05cm} \int_{-r_{j+1}x}^{-r_{j}x} \hspace{-0.05cm} \frac{x|s|^{\alpha}}{(s-xz)^{2}}ds = (r_{j}x)^{\alpha}.
\end{equation*}
For $j = k$, the analysis is slightly simpler, because only one term in the sum \eqref{def of f} contributes to the limit \eqref{lol 3}. Thus, we obtain as $x \to 0$
\begin{equation}\label{lol 2}
b_{j}(\sqrt{x}) = \frac{\sqrt{x}}{2\pi} P_{\mathrm{Be},0,11}^{2}(0;\alpha)(r_{j}x)^{\alpha}(s_{j+1}-s_{j})(1+\bigO(x)), \qquad j \in \{1,...,k\}.
\end{equation}
We will now use the explicit form of $P_{\mathrm{Be}}$ given in the appendix, see \eqref{def of model Bessel}. Since $P_{\mathrm{Be},0,11}(z;\alpha)$ is an entire function in $z$, we can obtain $P_{\mathrm{Be},0,11}(0;\alpha)$ by taking the limit $z\to 0$ from any region. In particular, for $z \in \{  z \in \mathbb{C}:|\arg(z)|<\frac{2\pi}{3} \}$, we have
\begin{equation}
P_{\mathrm{Be},0,11}(z;\alpha) = \sqrt{\pi}z^{-\frac{\alpha}{2}}I_{\alpha}(z^{\frac{1}{2}}).
\end{equation}
By using the behaviour of $I_{\alpha}(z)$ as $z \to 0$ (see \cite[formula 10.30.1]{NIST}) we obtain $P_{\mathrm{Be},0,11}(0;\alpha) = \frac{\sqrt{\pi}}{2^{\alpha}\Gamma(\alpha+1)}$. The equation \eqref{lol 2} can now be rewritten as
\begin{equation}\label{asymptotics behaviour as x to 0 for b_j}
b_{j}(\sqrt{x}) = \frac{\sqrt{x}(s_{j+1}-s_{j})}{2} J_{\alpha}\big(\sqrt{r_{j}x}\big)^{2}(1+\bigO(x)), \qquad \mbox{ as } x \to 0,
\end{equation}
for any $j \in \{1,...,k\}$. By the definition of $q_{j}$, see \eqref{def of q_j}, we have 
\begin{equation}
q_{j}(x) = \sqrt{s_{j+1}-s_{j}} J_{\alpha}(\sqrt{r_{j}x}), \qquad \mbox{ as } x \to 0,
\end{equation}
which is precisely the boundary conditions of the system \eqref{system of ODEs for the q_j}. In particular, the functions $q_{1}^{2}(x)$,...,$q_{k}^{2}(x)$ are integrable on $(0,\epsilon)$ for any $\epsilon >0$, and this proves \eqref{existence of the integral}. Also, \eqref{explicit expression for Phi in the small x} implies that as $x \to 0$ we have
\begin{equation*}
\begin{array}{r c l}
\displaystyle \lim_{z \to -r_{j}x}[\Phi^{-1}(z;\vec{r}x,\vec{s})\Phi^{\prime}(z;\vec{r}x,\vec{s})]_{21} & = & e^{-\pi i \alpha}(r_{j}x)^{\alpha} \left( [P_{\mathrm{Be},0}^{-1}(0;\alpha)P^{\prime}_{\mathrm{Be},0}(0;\alpha)]_{21} + \bigO(x) \right) \\
& = & \bigO(x^{\alpha}),
\end{array}
\end{equation*}
for every $j \in \{1,...,k\}$, and where the limit in taken from $z \in \mathcal{I}_{4}$. This proves \eqref{small eps limit phi}.
\section{Asymptotics for $s_{j} \to s_{j+1}$, $j \in \{1,...,k\}$}
\label{Section: asymptotics as s_j to s_j+1}
In this section, we perform a Deift/Zhou steepest descent \cite{DeiftZhou1992, DeiftZhou, DKMVZ2, DKMVZ1} to obtain asymptotics as $s_{j} \to s_{j+1}$ for $\Phi(z;\vec{r}x,\vec{s})$ uniformly in $z$, and where $\vec{r}$ and $\vec{s}$ satisfy conditions \eqref{condition r} and \eqref{condition s}. Let us fixed $j \in \{1,...,k\}$ in this section. If $j \neq 1$, we assume furthermore that $s_{j+1}\neq s_{j-1}$. When $s_{j} \to s_{j+1}$, the jumps of $\Phi$ on $(-r_{j}x,-r_{j-1}x)$ tend to be the same as those on $(-r_{j+1}x,-r_{j}x)$ and therefore the logarithmic singularity at $z=-r_{j}x$ for $\Phi(z;\vec{r}x,\vec{s})$ tends to disappear. Consider $U_{j}$, a fixed open neighbourhood of $[-r_{j}x,-r_{j-1}x]$ with smooth boundaries, sufficiently small such that it does not include any $-r_{\ell}x$, $\ell \neq j$, $\ell \neq j-1$. Outside $U_{j}$, the model RH problem $\Phi(z;\vec{r}^{[j]}x,\vec{s}^{[j]})$ possesses exactly the same jumps and the same large $z$ asymptotics than $\Phi(z;\vec{r}x,\vec{s})$, and thus heuristically it is a good approximation of $\Phi(z;\vec{r}x,\vec{s})$ for $z \in \mathbb{C}\setminus U_{j}$. Furthermore, for $z \in U_{j}$, from \eqref{Phi x_j} and \eqref{Phi 0}, $\Phi(z;\vec{r}^{[j]}x,\vec{s}^{[j]})$ can be written as
\begin{equation}\label{lolo 1}
\hspace{-0.2cm}\Phi(z;\vec{r}^{[j]}x,\vec{s}^{[j]}) \hspace{-0.06cm} =  \hspace{-0.06cm}
\Phi_{0,j-1}^{\star}(z) \begin{pmatrix}
1 & \hspace{-0.15cm} \frac{s_{j+1}-s_{j-1}}{2\pi i} \log(z+r_{j-1}x) \\ 0 & \hspace{-0.15cm} 1
\end{pmatrix} \hspace{-0.08cm} V_{j-1}(z)e^{\frac{\pi i \alpha}{2}\theta(z)\sigma_{3}} H_{-r_{k}x}(z) 
\end{equation}
if $j \in \{2,...,k\}$, and as
\begin{equation}\label{lolo 2}
\Phi(z;\vec{r}^{[j]}x,\vec{s}^{[j]}) = \Phi_{0,0}^{\star}(z)z^{\frac{\alpha}{2}\sigma_{3}} \begin{pmatrix}
1 & s_{2}h(z) \\ 0 & 1
\end{pmatrix}, \qquad \mbox{ if } j = 1.
\end{equation}
Therefore, we define the local parametrix inside $U_{j}$ by
\begin{multline}\label{lolo 3}
P(z) =
\Phi_{0,j-1}^{\star}(z) \begin{pmatrix}
1 & \frac{s_{j+1}-s_{j}}{2\pi i} \log(z+r_{j}x) + \frac{s_{j}-s_{j-1}}{2\pi i} \log(z+r_{j-1}x) \\ 0 & 1
\end{pmatrix} \\ \times V_{j-1}(z)e^{\frac{\pi i \alpha}{2}\theta(z)\sigma_{3}} H_{-r_{k}x}(z),
\end{multline}
if $j \in \{2,...,k\}$, and by
\begin{equation}\label{lolo 4}
P(z) = \Phi_{0,0}^{\star}(z) \begin{pmatrix}
1 & \widetilde{f}(z;x)  \\ 0 & 1
\end{pmatrix} z^{\frac{\alpha}{2}\sigma_{3}} \begin{pmatrix}
1 & s_{2}h(z) \\ 0 & 1
\end{pmatrix}, \qquad \mbox{ if } j = 1,
\end{equation}
where $\widetilde{f}(z;x) = -\frac{s_{2}-s_{1}}{2\pi i}\int_{-r_{1}x}^{0} \frac{|s|^{\alpha}}{s-z}ds$. 
It is direct to check that $P(z)$ has exactly the same jumps as $\Phi(z;\vec{r}x,\vec{s})$ inside $U_{j}$. We define
\begin{equation}\label{def of R simple asymptotics}
R(z) = \left\{ \begin{array}{l l}
\Phi(z;\vec{r}x,\vec{s}) \Phi(z;\vec{r}^{[j]}x,\vec{s}^{[j]})^{-1}, & \mbox{ for } z \in \mathbb{C}\setminus U_{j}, \\
\Phi(z;\vec{r}x,\vec{s}) P(z)^{-1}, & \mbox{ for } z \in U_{j}.
\end{array} \right.
\end{equation}
From the above remarks, it follows that $R(z) = I + \bigO(z^{-1})$ as $z \to \infty$ and $R(z)$ has no jump at all inside and outside $U_{j}$, and has removable singularities at $0$,$-x_{1}$,...,$-x_{k}$. Let us denote the boundaries of $U_{j}$ by $\partial U_{j}$, whose orientation is chosen to be clockwise. For $z \in \partial U_{j}$, the jumps $J_{R}$ of $R$ satisfy $J_{R}(z) = P(z)\Phi(z;x,\vec{r}^{[j]},\vec{s}^{[j]})^{-1}$, or more explicitly
\begin{equation}\label{jumps for R simple asymptotics}
\hspace{-0.3cm}J_{R}(z) = \left\{ \hspace{-0.12cm} \begin{array}{l l}
\Phi_{0,j-1}^{\star}(z) \begin{pmatrix}
1 & \frac{s_{j+1}-s_{j}}{2\pi i}\log\left(\frac{z+r_{j}x}{z+r_{j-1}x}\right) \\
0 & 1
\end{pmatrix}\Phi_{0,j-1}^{\star}(z)^{-1}, & \mbox{if } j \in \{2,...,k\}, \\
\Phi_{0,0}^{\star}(z) \begin{pmatrix}
1 & \widetilde{f}(z;x)\\
0 & 1
\end{pmatrix}\Phi_{0,0}^{\star}(z)^{-1}, & \mbox{if } j = 1.
\end{array} \right.
\end{equation}
In all the cases, we thus have $J_{R}(z) = I + \bigO(s_{j+1}-s_{j})$ as $s_{j+1}-s_{j} \to 0$ uniformly for $z \in \partial U_{j}$. It follows from standard analysis for small norm RH problems that $R$ exists for sufficiently small $s_{j+1}-s_{j}$ and satisfies
\begin{equation}\label{asymptotics small norm R s_j regime}
R(z) = I + \bigO(s_{j+1}-s_{j}), \qquad R^{\prime}(z) = \bigO(s_{j+1}-s_{j}),
\end{equation}
uniformly for $z \in \mathbb{C}\setminus  \partial U_{j}$. We now turn to the small $s_{j+1}-s_{j}$ asymptotics for $q_{1}$,...,$q_{k}$. For convenience, we rewrite \eqref{lol 3}, but we explicit the dependence in $\vec{r}$ and in $\vec{s}$:
\begin{equation}\label{lol 4}
b_{\ell}(\sqrt{x};\vec{r},\vec{s}) = i\sqrt{x} \lim_{z\to -r_{\ell}}(z+r_{\ell}) \left[ \partial_{z} (\Phi(xz;\vec{r}x,\vec{s}))\Phi^{-1}(xz;\vec{r}x,\vec{s}) \right]_{12},
\end{equation}
for any $\ell \in \{1,...,k\}$.

\hspace{-0.6cm}If $\ell \neq j$ and $\ell \neq j-1$, then $-r_{\ell}x \notin U_{j}$, and in the above limit, from \eqref{def of R simple asymptotics} we have to use $\Phi(xz;\vec{r}x,\vec{s}) = R(xz) \Phi(xz;\vec{r}^{[j]}x,\vec{s}^{[j]})$, and thus
\begin{multline}\label{lol 5}
\partial_{z} (\Phi(xz;\vec{r}x,\vec{s}))\Phi^{-1}(xz;\vec{r}x,\vec{s}) = \partial_{z}(R(xz))R(xz)^{-1} \\ + R(xz) \partial_{z} (\Phi(xz;\vec{r}^{[j]}x,\vec{s}^{[j]}))\Phi^{-1}(xz;\vec{r}^{[j]}x,\vec{s}^{[j]}) R(xz)^{-1}.
\end{multline}
By \eqref{asymptotics small norm R s_j regime} and \eqref{lol 4}, this implies
\begin{align}
& b_{\ell}(\sqrt{x};\vec{r},\vec{s}) = b_{\tilde{\ell}}(\sqrt{x};\vec{r}^{[j]},\vec{s}^{[j]}) +\bigO(s_{j+1}-s_{j}), \qquad \ell \notin\{ j,j-1 \}, 
\end{align}
where $\tilde{\ell} = \ell$ if $\ell < j-1$ and $\tilde{\ell} = \ell-1$ if $\ell > j$.
If $\ell \in \{j,j-1\}$, then $-r_{\ell}x \in U_{j}$ and we have to use the local parametrix:
\begin{multline}\label{lol 6}
b_{\ell}(\sqrt{x},\vec{r},\vec{s}) = i\sqrt{x} \lim_{z\to -r_{\ell}}(z+r_{\ell}) \left[ \partial_{z}(R(xz))R(xz)^{-1} \right. \\ \left. + R(xz)\partial_{z}(P(xz))P^{-1}(xz)R^{-1}(xz) \right]_{12}.
\end{multline}
Note from \eqref{lolo 1} and \eqref{lolo 3} that for $j \in \{2,...,k\}$ we have as $z \to -r_{j-1}$ that
\begin{equation*}
\frac{\displaystyle [\partial_{z}(P(xz))P^{-1}(xz)]_{12}}{\displaystyle [\partial_{z} (\Phi(xz;\vec{r}^{[j]}x,\vec{s}^{[j]}))\Phi^{-1}(xz;\vec{r}^{[j]}x,\vec{s}^{[j]})]_{12}} \sim \left( 1- \frac{s_{j+1}-s_{j}}{s_{j+1}-s_{j-1}} \right),
\end{equation*}
and for $j = 1$, from \eqref{lolo 2} and \eqref{lolo 4}, we have as $z \to 0$ that
\begin{equation*}
\frac{\displaystyle [\partial_{z}(P(xz))P^{-1}(xz)]_{12}}{\displaystyle [\partial_{z} (\Phi(xz;\vec{r}^{[j]}x,\vec{s}^{[j]}))\Phi^{-1}(xz;\vec{r}^{[j]}x,\vec{s}^{[j]})]_{12}} \sim \frac{\displaystyle [\Phi_{0,0}^{\star}(0)\begin{pmatrix}
1 & \widetilde{f}(0) \\ 0 & 1
\end{pmatrix}\sigma_{3}\begin{pmatrix}
1 & -\widetilde{f}(0) \\ 0 & 1
\end{pmatrix}\Phi_{0,0}^{\star}(0)^{-1}]_{12}}{\displaystyle [\Phi_{0,0}^{\star}(0)\sigma_{3}\Phi_{0,0}^{\star}(0)^{-1}]_{12}}.
\end{equation*}
Thus, we also obtain $b_{j-1}(\sqrt{x};\vec{r},\vec{s}) = b_{j-1}(\sqrt{x};\vec{r}^{[j]},\vec{s}^{[j]})+\bigO(s_{j+1}-s_{j})$ as $s_{j}\to s_{j+1}$. When $\ell = j \in \{2,...,k\}$, from \eqref{lolo 3} and \eqref{lol 4}, we have as $s_{j} \to s_{j+1}$
\begin{equation}
b_{j}(\sqrt{x};\vec{r},\vec{s}) = i\sqrt{x} \frac{s_{j+1}-s_{j}}{2\pi i}  (\Phi_{0,j-1}^{\star}(-r_{j}x))_{11}^{2}(1+\bigO(s_{j+1}-s_{j})) = \bigO(s_{j+1}-s_{j}).
\end{equation}
For $\ell=j=1$, from \eqref{lolo 4}, we obtain similarly that
\begin{equation}\label{lol 19}
b_{1}(\sqrt{x};\vec{r},\vec{s}) = i\sqrt{x}\frac{s_{2}-s_{1}}{2\pi i}(r_{1}x)^{\alpha}(\Phi_{0,0}^{\star}(-r_{1}x))_{11}^{2}(1+\bigO(s_{2}-s_{1}))= \bigO(s_{2}-s_{1}).
\end{equation}
This finishes the proof of part 1 of Theorem \ref{thm: residual thm}.

\section{Asymptotics for $r_{j} \to r_{j-1}$, $j \in \{1,...,k\}$}
\label{Section: Asymptotics as r_j to r_j-1}
As $r_{j} \to r_{j-1}$, the jumps along $(-r_{j}x,-r_{j-1}x)$ tends to disappear. Thus, we do the exactly the same steepest descent as in the previous section. The computations are very similar and we will give less details. We define $R$ exactly as in \eqref{def of R simple asymptotics}. By \eqref{jumps for R simple asymptotics}, we have $J_{R}(z) = \bigO(r_{j}-r_{j-1})$ as $r_{j} \to r_{j-1}$ uniformly for $z \in \partial U_{j}$. Thus, by standard theory for small norm RH problems, $R$ exists for sufficiently small $r_{j}-r_{j-1}$ and satisfies
\begin{equation}\label{asymp R merging r}
R(z) = I + \bigO(r_{j}-r_{j-1}), \qquad R^{\prime}(z) = \bigO(r_{j}-r_{j-1}),
\end{equation}
uniformly for $z \in \mathbb{C}\setminus  \partial U_{j}$. From \eqref{lol 4}, \eqref{lol 5} together with \eqref{asymp R merging r}, we have 
\begin{equation}
b_{\ell}(\sqrt{x};\vec{r},\vec{s}) = b_{\ell}(\sqrt{x};\vec{r}^{[j]},\vec{s}^{[j]})+\bigO(r_{j}-r_{j-1}), \quad \mbox{ for any } \ell \in \{1,...,k\}\setminus\{ j,j-1 \},
\end{equation}
Let $j \in \{2,...,k\}$. From \eqref{lolo 3}, \eqref{lol 6} and \eqref{asymp R merging r}, we obtain
\begin{align}
& b_{j-1}(\sqrt{x};\vec{r},\vec{s}) = \frac{s_{j}-s_{j-1}}{s_{j+1}-s_{j-1}} b_{j-1}(\sqrt{x};\vec{r}^{[j]},\vec{s}^{[j]})+\bigO(r_{j}-r_{j-1}), \\
& b_{j}(\sqrt{x};\vec{r},\vec{s}) = \frac{s_{j+1}-s_{j}}{s_{j+1}-s_{j-1}}b_{j-1}(\sqrt{x};\vec{r}^{[j]},\vec{s}^{[j]})+\bigO(r_{j}-r_{j-1}).
\end{align}
This proves part 2 of Theorem \ref{thm: residual thm}.
If $j = 1$, the computations are very similar to \eqref{lol 19}. From \eqref{lolo 4}, \eqref{lol 6} and \eqref{asymp R merging r}, we obtain
\begin{equation}
b_{1}(\sqrt{x};\vec{r},\vec{s}) = i\sqrt{x}\frac{s_{2}-s_{1}}{2\pi i}(r_{1}x)^{\alpha}(\Phi_{0,0}^{\star}(-r_{1}x))_{11}^{2}(1+\bigO(r_{1}))= \bigO(r_{1}^{\alpha}),
\end{equation}
which is the part 3 of Theorem \ref{thm: residual thm}.

\appendix
\section{Bessel model RH problem}\label{Appendix: Bessel}
In this appendix, we present the well-known Bessel model RH problem, whose solution is denoted  $P_{\mathrm{Be}}$ and depends on a parameter $\alpha >-1$. At the end of the appendix, we also define $\widetilde{P}_{\mathrm{Be}}$, which a obtained from $P_{\mathrm{Be}}$ by a simple transformation and satisfied a modified version of the Bessel model RH problem.
\subsection*{RH problem for $P_{\mathrm{Be}}( z ) = P_{\mathrm{Be}}( z ;\alpha)$}
\begin{itemize}
\item[(a)] $P_{\mathrm{Be}} : \mathbb{C} \setminus \Sigma_{B} \to \mathbb{C}^{2\times 2}$ is analytic, where
$\Sigma_{B}$ is shown in Figure \ref{figBessel}.
\item[(b)] $P_{\mathrm{Be}}$ satisfies the jump conditions
\begin{equation}\label{Jump for P_Be}
\begin{array}{l l} 
P_{\mathrm{Be},+}( z ) = P_{\mathrm{Be},-}( z ) \begin{pmatrix}
0 & 1 \\ -1 & 0
\end{pmatrix}, &  z \in \mathbb{R}^{-}, \\

P_{\mathrm{Be},+}( z ) = P_{\mathrm{Be},-}( z ) \begin{pmatrix}
1 & 0 \\ e^{\pi i \alpha} & 1
\end{pmatrix}, &  z  \in e^{\frac{2\pi i}{3} }\mathbb{R}^{+}, \\

P_{\mathrm{Be},+}( z ) = P_{\mathrm{Be},-}( z ) \begin{pmatrix}
1 & 0 \\ e^{-\pi i \alpha} & 1
\end{pmatrix}, &  z  \in e^{-\frac{2\pi i}{3} }\mathbb{R}^{+}. \\
\end{array}
\end{equation}
\item[(c)] As $ z \to \infty$, $ z  \notin \Sigma_{B}$, we have
\begin{equation}\label{large z asymptotics Bessel}
P_{\mathrm{Be}}( z ) =  
\left(I+\bigO( z ^{-1})\right)  z ^{\frac{-\sigma_{3}}{4}} N e^{ z ^{\frac{1}{2}}\sigma_{3}}.
\end{equation}
\item[(d)] As $ z $ tends to 0, the behaviour of $P_{\mathrm{Be}}( z )$ is
\begin{equation}\label{local behaviour near 0 of P_Be}
\begin{array}{l l}
\displaystyle P_{\mathrm{Be}}( z ) = \left\{ \begin{array}{l l}
\bigO(1) z ^{\frac{\alpha}{2}\sigma_{3}}, & |\arg  z  | < \frac{2\pi}{3}, \\
\bigO( z ^{-\frac{\alpha}{2}}), & \frac{2\pi}{3}<|\arg  z  | < \pi,
\end{array} \right. , & \displaystyle \mbox{ if } \alpha > 0, \\[0.3cm]
\displaystyle P_{\mathrm{Be}}( z ) = \bigO(\log  z ), & \displaystyle \mbox{ if } \alpha = 0, \\[0.3cm]
\displaystyle P_{\mathrm{Be}}( z ) = \bigO( z ^{\frac{\alpha}{2}}), & \displaystyle \mbox{ if } \alpha < 0.
\end{array}
\end{equation}
\end{itemize}
Note that by deleting the jumps of $P_{\mathrm{Be}}$ around the origin, we obtain the relation
\begin{equation}\label{local Bessel behaviour}
P_{\mathrm{Be}}( z ) = P_{\mathrm{Be},0}( z ) z ^{\frac{\alpha}{2}\sigma_{3}}\begin{pmatrix}
1 & h( z ) \\ 0 & 1
\end{pmatrix} H_{0}( z ), \qquad z \in \mathbb{C}\setminus \Sigma_{B},
\end{equation}
where $P_{\mathrm{Be},0}$ is an entire function, $h$ is defined in \eqref{def of h} and $H_{0}$ is defined in \eqref{def of H}.
\begin{figure}[t]
    \begin{center}
    \setlength{\unitlength}{1truemm}
    \begin{picture}(100,55)(-5,10)
        \put(50,40){\line(-1,0){30}}
        \put(50,39.8){\thicklines\circle*{1.2}}
        \put(50,40){\line(-0.5,0.866){15}}
        \put(50,40){\line(-0.5,-0.866){15}}
        \put(50.3,36.8){$0$}
        \put(35,39.9){\thicklines\vector(1,0){.0001}}
        \put(41,55.588){\thicklines\vector(0.5,-0.866){.0001}}
        \put(41,24.412){\thicklines\vector(0.5,0.866){.0001}}
    \end{picture}
    \caption{\label{figBessel}The jump contour $\Sigma_{B}$ for $P_{\mathrm{Be}}( z )$.}
\end{center}
\end{figure}
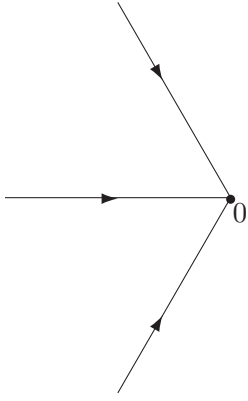
It was shown in \cite{KMcLVAV} that the unique solution to this RH problem is given by 
\begin{equation}\label{def of model Bessel}
P_{\mathrm{Be}}( z )= \begin{pmatrix}
1 & 0 \\ \frac{i}{8}(4\alpha^{2}+3) & 1
\end{pmatrix} \pi^{\frac{\sigma_{3}}{2}} \begin{pmatrix}
I_{\alpha}( z ^{\frac{1}{2}}) & \frac{ i}{\pi} K_{\alpha}( z ^{\frac{1}{2}}) \\
\pi i  z ^{\frac{1}{2}} I_{\alpha}^{\prime}( z ^{\frac{1}{2}}) & -  z ^{\frac{1}{2}} K_{\alpha}^{\prime}( z ^{\frac{1}{2}})
\end{pmatrix} H_{0}( z ).
\end{equation}
where $I_\alpha$ and $K_\alpha$ are the modified Bessel functions of the first and second kind. 

\hspace{-0.6cm}Note that $\begin{pmatrix}
I_{\alpha}( z ^{\frac{1}{2}}) & \frac{ i}{\pi} K_{\alpha}( z ^{\frac{1}{2}}) \\
\pi i  z ^{\frac{1}{2}} I_{\alpha}^{\prime}( z ^{\frac{1}{2}}) & -  z ^{\frac{1}{2}} K_{\alpha}^{\prime}( z ^{\frac{1}{2}})
\end{pmatrix} H_{0}( z )$ can be rewritten as
\begin{equation}
\begin{cases}
\begin{pmatrix}
I_{\alpha}( z ^{\frac{1}{2}}) & \frac{ i}{\pi} K_{\alpha}( z ^{\frac{1}{2}}) \\
\pi i  z ^{\frac{1}{2}} I_{\alpha}^{\prime}( z ^{\frac{1}{2}}) & -  z ^{\frac{1}{2}} K_{\alpha}^{\prime}( z ^{\frac{1}{2}})
\end{pmatrix}, & \mbox{ if } |\arg  z  | < \frac{2\pi}{3}, \\

\begin{pmatrix}
\frac{1}{2} H_{\alpha}^{(1)}((- z )^{\frac{1}{2}}) & \frac{1}{2} H_{\alpha}^{(2)}((- z )^{\frac{1}{2}}) \\
\frac{1}{2}\pi  z ^{\frac{1}{2}} \left( H_{\alpha}^{(1)} \right)^{\prime} ((- z )^{\frac{1}{2}}) & \frac{1}{2}\pi  z ^{\frac{1}{2}} \left( H_{\alpha}^{(2)} \right)^{\prime} ((- z )^{\frac{1}{2}})
\end{pmatrix}e^{\frac{\pi i \alpha}{2}\sigma_{3}}, & \mbox{ if }\frac{2\pi}{3} < \arg  z  < \pi, \\

\begin{pmatrix}
\frac{1}{2} H_{\alpha}^{(2)}((- z )^{\frac{1}{2}}) & -\frac{1}{2} H_{\alpha}^{(1)}((- z )^{\frac{1}{2}}) \\
-\frac{1}{2}\pi  z ^{\frac{1}{2}} \left( H_{\alpha}^{(2)} \right)^{\prime} ((- z )^{\frac{1}{2}}) & \frac{1}{2}\pi  z ^{\frac{1}{2}} \left( H_{\alpha}^{(1)} \right)^{\prime} ((- z )^{\frac{1}{2}})
\end{pmatrix}e^{-\frac{\pi i \alpha}{2}\sigma_{3}}, & \mbox{ if } -\pi < \arg  z  < -\frac{2\pi}{3},
\end{cases}
\end{equation}
where $H_{\alpha}^{(1)}$ and $H_{\alpha}^{(2)}$ are the Hankel functions of the first and second kind respectively.

\vspace{0.2cm}\hspace{-0.6cm}We will also use a modified version of the above RH problem. We define
\begin{equation}\label{def of P_Be tilde}
\widetilde{P}_{\mathrm{Be}}( z ) = e^{-\frac{\pi i}{4}\sigma_{3}}\pi^{\frac{\sigma_{3}}{2}} \begin{pmatrix}
I_{\alpha}( z ^{\frac{1}{2}}) & \frac{ i}{\pi} K_{\alpha}( z ^{\frac{1}{2}}) \\
\pi i  z ^{\frac{1}{2}} I_{\alpha}^{\prime}( z ^{\frac{1}{2}}) & -  z ^{\frac{1}{2}} K_{\alpha}^{\prime}( z ^{\frac{1}{2}})
\end{pmatrix} H_{-x_{k}}( z ).
\end{equation}
From \eqref{def of model Bessel}, we have that $\widetilde{P}_{\mathrm{Be}}$ has exactly the same jumps than $P_{\mathrm{Be}}$ on $\Sigma_{1} \cup \Sigma_{2} \cup (-\infty,-x_{k})$. We can compute the jumps of $\widetilde{P}_{\mathrm{Be}}$ on $(-x_{k},0)$ either from the properties of the Bessel functions, or from the jumps of $P_{\mathrm{Be}}$. We obtain
\begin{equation}
\widetilde{P}_{\mathrm{Be},+}( z ) = \widetilde{P}_{\mathrm{Be},-}(z) \begin{pmatrix}
e^{\pi i \alpha} & 1 \\ 0 & e^{-\pi i \alpha}
\end{pmatrix}, \qquad z \in (-x_{k},0).
\end{equation}
Also, from \eqref{large z asymptotics Bessel}, as $ z \to \infty$, $ z  \notin \Sigma_{\Phi}$, we have
\begin{equation}\label{large z asymptotics modified Bessel}
\widetilde{P}_{\mathrm{Be}}( z ) =  e^{-\frac{\pi i}{4}\sigma_{3}}
\begin{pmatrix}
1 & 0 \\ \frac{-i}{8}(4\alpha^{2}+3) & 1
\end{pmatrix}\left(I+\bigO( z ^{-1})\right)  z ^{\frac{-\sigma_{3}}{4}} N e^{ z ^{\frac{1}{2}}\sigma_{3}}.
\end{equation}

\section*{Acknowledgements}
C. Charlier was supported by the European Research Council under the European Union's Seventh Framework Programme (FP/2007/2013)/ ERC Grant Agreement n.\, 307074. Both authors also acknowledge support by the Belgian Interuniversity Attraction Pole P07/18. We acknowledge the anonymous referee for a careful reading and for useful remarks.

\end{document}